\documentclass[12pt]{article}

\usepackage{colonequals}
\usepackage{graphicx}
\usepackage{amssymb}
\usepackage{epstopdf}
\usepackage{graphicx}
\usepackage{subfigure}
\usepackage{pstricks-add}
\usepackage{pst-slpe}
\usepackage{color}
\usepackage{amsthm}
\usepackage{authblk}
\usepackage{placeins}
\usepackage[small]{caption}
\usepackage{algorithmic}
\usepackage{algorithm}
\usepackage{amsmath}
\usepackage{wasysym}
\usepackage[bookmarks, colorlinks=true, linkcolor=blue, citecolor=green]{hyperref}

\title{Quantification and Minimization of Crosstalk Sensitivity in Networks}
\author[1,2]{Dionysios Barmpoutis}
\author[2]{Richard M. Murray}
\affil[1]{Computation and Neural Systems}
\affil[2]{Control and Dynamical Systems \bigskip}
\affil[ ]{California Institute of Technology\bigskip}
\affil[ ]{dionysios$@$caltech.edu}
\affil[ ]{murray$@$cds.caltech.edu}
\date{\today}                                        

\newtheorem{lemma}{Lemma}
\newtheorem{theorem}{Theorem}
\newtheorem{corollary}{Corollary}

\theoremstyle{definition}
\newtheorem{definition}{Definition}

\DeclareMathOperator*{\argmin}{arg\,min}

\begin{document}
\maketitle
\begin{abstract}
Crosstalk is defined as the set of unwanted interactions among the different entities of a network.
Crosstalk is present in various degrees in every system where information is transmitted through a means that is accessible by all the individual units of the network.
Using concepts from graph theory, we introduce a quantifiable measure for sensitivity to crosstalk, and analytically derive the structure of the networks in which it is minimized.
It is shown that networks with an inhomogeneous degree distribution are more robust to crosstalk than corresponding homogeneous networks.
We provide a method to construct the graph with the minimum possible sensitivity to crosstalk, given its order and size.
Finally, for networks with a fixed degree sequence, we present an algorithm to find the optimal interconnection structure among their vertices.
\end{abstract}
%%%%%%%%%%%%%%%%%%%%%%%%%%%%%%%%%%%%%%%%%%%%%%%%
%%%%%%%%%%%%%%%%%%%%%%%%%%%%%%%%%%%%%%%%%%%%%%%%
%%%%%%%%%%%%%%%%%%%%%%%%%%%%%%%%%%%%%%%%%%%%%%%%
\section{Introduction}

Crosstalk affects the function of many complex engineering systems.
In microelectronic circuits, as the frequencies increase, it is increasingly hard to ensure proper function without isolating the various components from electromagnetic interference.
In wireless communications, we need to make sure that the input and output signal frequency spectrums are relatively narrow to avoid crosstalk with other communication systems.
In chemical and biomolecular systems, different molecules interact with each other in a solution without explicit mechanisms that bring them in contact.
As a result, there is a potentially large number of spurious interactions, despite which the cell can function without problems.
The potential amount of unwanted crosstalk interactions increases as a function of the number of the individual elements in the network, so there is the need to design systems that are optimized to minimize or at least decrease the amount of spurious interactions, in order to keep the network functional.
In biological networks, this need is more pronounced since signaling is implemented through variations of molecular concentrations.
Crosstalk does not seem to be a problem in the cell, or any other natural biological system, despite the large variety of molecules, and the intricate patterns of interactions. 
However, in biological design, even the simplest systems suffer from crosstalk between different elements, notwithstanding the small complexity of such artificial systems \cite{PurnickWeiss2009}.
If there are molecules that randomly interact with each other without performing any function, the effective concentration of each molecule which is available to participate in any given pathway will tend to decrease, and therefore the efficiency of the circuit will diminish.

In this paper, we find the structure of graphs with the smallest possible sensitivity to crosstalk, assuming that the crosstalk affinity of each vertex is either specific to each unit in the network and independent of anything else, or is a function of its degree, namely the number of other units with which it is designed to interact within the network.
%%%%%%%%%%%%%%%%%%%%%%%%%%%%%%%%%%%%%%%%%%%%%%%%
%%%%%%%%%%%%%%%%%%%%%%%%%%%%%%%%%%%%%%%%%%%%%%%%
%%%%%%%%%%%%%%%%%%%%%%%%%%%%%%%%%%%%%%%%%%%%%%%%
\section{Graph Theory Preliminaries}
This section provides a brief introduction to the notions needed from graph theory that are used throughout this study.
A \textit{graph} (also called a \textit{network}) is an ordered pair $\mathcal{G}=(\mathcal{V},\mathcal{E})$ comprised of a set $\mathcal{V}=\mathcal{V}(\mathcal{G})$ of \textit{vertices} together with a set $\mathcal{E}=\mathcal{E}(\mathcal{G})$ of \textit{edges} that are unordered 2-element subsets of $\mathcal{V}$.
Two vertices $u$ and $v$ are called \textit{neighbors} if they are connected through an edge ($(u,v)\in \mathcal{E}$) and we write $u-v$, otherwise we write $u \notslash v$.
The \textit{neighborhood} $\mathcal{N}_{u}$ of a vertex $u$ is the set its neighbors.
The \textit{degree}  of a vertex is the number of neighbors it has.
A vertex is said to have \textit{full degree} if it is connected to every other vertex in the network.
A network is \textit{assortative} with respect to its degree distribution when the vertices in the network are connected to others that have similar degrees.
All graphs in this article are \textit{simple}, meaning that all edges connect two different vertices, and there is no more than one edge between any pair of vertices.
The \textit{order} of a graph is the number of its vertices, $|\mathcal{V}|$.
A graph's \textit{size} is $|\mathcal{E}|$, the number of edges.
We will denote a graph $\mathcal{G}$ of order $N$ and size $m$ as $\mathcal{G}(N,m)$ or simply $\mathcal{G}_{N,m}$.
A \textit{complete graph} is a graph in which each vertex is connected to every other.
A \textit{clique} in a graph is a subset of its vertices such that every two vertices in the subset are connected.
The \textit{order} of a clique is the number of vertices that belong to it.
A \textit{graphic sequence} is a sequence of numbers which can be the degree sequence of a graph. 
A \textit{path} is a sequence of consecutive edges in a graph and the length of the path is the number of edges traversed.
A graph is \textit{connected} if for every pair of vertices $u$ and $v$, there is a path from $u$ to $v$.
Otherwise the graph is called \textit{disconnected}.
A \textit{tree} is a graph in which any two vertices are connected by exactly one path.
A weighted graph associates a weight with every edge.
Weights will be restricted to positive real numbers.
An edge is \textit{rewired} when we change the vertices it is adjacent to.
A \textit{single rewiring} takes place when we change one of the vertices that is adjacent to it, and a \textit{double rewiring} when we change both of them.
Finally, a subgraph $\mathcal{H}$ of a graph $\mathcal{G}$ is called \textit{induced} if $\mathcal{V}(\mathcal{H}) \subseteq \mathcal{V}(\mathcal{G})$ and for any pair of vertices $u$ and $v$ in $\mathcal{V}(\mathcal{H})$, $(u,v)\in \mathcal{E}(\mathcal{H})$ if and only if $(u,v) \in \mathcal{E}(\mathcal{G})$.
 In other words, $\mathcal{H}$  is an induced subgraph of $\mathcal{G}$ if it has the same edges that appear in $\mathcal{G}$ over the same vertex set. 
 
%%%%%%%%%%%%%%%%%%%%%%%%%%%%%%%%%%%%%%%%%%%%%%%%
%%%%%%%%%%%%%%%%%%%%%%%%%%%%%%%%%%%%%%%%%%%%%%%%
%%%%%%%%%%%%%%%%%%%%%%%%%%%%%%%%%%%%%%%%%%%%%%%%
\section{Model}
\subsection{General Considerations}
Assume that we have a weighted graph that represents a network of interactions.
We define crosstalk as additional unmodeled interactions among the vertices of the network.
These edges are randomly distributed across the network, and typically have smaller weights than the original ones.
Throughout the paper, we denote the order of the network $|\mathcal{V}|$ as $N$, and its size $|\mathcal{E}|$ as $m$.
We require that there cannot be any unwanted interactions among vertices that are already connected in the original network.
The crosstalk matrix, denoted by $R^{xt}$, is the weight matrix of all the crosstalk interactions.
If there is crosstalk between vertices $u$ and $v$, the matrix element $w^{xt}_{u,v}$ will contain the intensity of that interaction.
We will assume that the crosstalk affinity between any pair of vertices $(u,v) \notin \mathcal{E}(\mathcal{G})$, is bounded:
\begin{equation}
||R^{xt}||_{\infty}= \max _{(u,v) \notin \mathcal{E}} w^{xt}_{u,v} \leq \epsilon.
\end{equation}
\begin{definition}
Given a local measure of the crosstalk interaction $w^{xt}_{u,v}=g(u,v)$ among two vertices $u$ and $v$, the \textit{overall sensitivity to crosstalk} of a network of interactions is defined as the sum of the intensities of all the spurious interactions among individual units that are not connected,
\begin{equation}
K(\mathcal{G})=\sum _{(u,v) \notin \mathcal{E}} w^{xt}_{u,v}.
\end{equation}
\end{definition}
When the network in question is obvious from the context, we will refer to its  sensitivity to crosstalk simply as $K$.
The order of the graph is finite, so the overall crosstalk sensitivity is bounded.
\begin{equation}
||R^{xt}|| _{1}= \sum _{(u,v) \notin \mathcal{E}(\mathcal{G})} w_{u,v}^{xt} \leq \left( {N \choose 2} -m\right)\epsilon =\mathcal{M}.
\end{equation}
We will call $\mathcal{M}=\mathcal{M}(N,m,\epsilon)$ the \textit{maximum possible crosstalk intensity} of the network.
It is only a function of the size and order of the network along with the maximum possible crosstalk intensity of any pair of vertices, which depends on the network and is independent of its specific structure.
The matrix $R^{xt}$ is constant, assuming that each crosstalk interaction represents the average over a time period.
\begin{definition}
We say that a graph $\mathcal{G}^{*}=\mathcal{G}^{*}_{N,m}$ has minimum sensitivity to crosstalk, or that it is a \textit{minimum crosstalk graph}, if for any other graph $\mathcal{G}=\mathcal{G}_{N,m}$ of the same order and size, we have
\begin{equation}
K(\mathcal{G}^{*}) \leq K(\mathcal{G}).
\end{equation}
\end{definition}
In what follows, we analyze the crosstalk sensitivity of networks in three different cases.
The first is when the crosstalk affinity of each vertex  depends only on its physical properties, specific to each network.
The second and third cases assume that the crosstalk affinity of each vertex is an increasing function of its degree.
This is a reasonable assumption, especially in biological networks, natural or man-made.
Biological molecules (proteins, for example) which are supposed to interact with many others, have a less specific shape.
This is exactly the reason why they can also randomly interact with many other unrelated molecules.
The same argument holds for molecules that are very specific, and can interact with a small subset of the rest of the molecules in the network.
In DNA systems, if a strand is designed to interact with many others, it is usually longer and, as a result it can stick to many unrelated strands that are floating in the solution.
In communication networks, transceivers that are sensitive to a large spectrum of frequencies are more likely to pick up noise or other random signals from surrounding transmitters, and the opposite is true for transceivers which have a narrow input and output frequency range.
 
\subsection{Crosstalk Specific to Individual Vertices}

In this first case, we assume that the amount of crosstalk depends on the specific vertices of the network and the way they interact with each other.
For this reason, we have no control over their intensity.
For a given number of vertices, the larger the size of the graph, the fewer the pairs of vertices that can spuriously interact with each other, and the overall sensitivity to crosstalk is decreased.
Two extreme cases are the tree graph, when the number of edges is equal to the minimum possible ($m=N-1$, given that the graph is required to be connected) and the complete graph, where $m={N \choose 2}$.
In the tree graph we have ${N \choose 2}-(N-1) ={N-1\choose 2}$ pairs of vertices that may spuriously interact.
The complete graph has no crosstalk interactions, simply because every vertex is designed to interact with every other.
In a network with size $m$, the overall crosstalk sensitivity $K(\mathcal{G})$ will be bounded by these two extremes,
\begin{equation}
0 \leq  K \leq {N-1 \choose 2} ||R^{xt}||_{\infty}.
\end{equation}
We can optimize the structure of this network, by connecting the $m$ pairs of vertices with the strongest interactions among each other, thus reducing $||R^{xt}||_{\infty}$. 

A more interesting problem arises when the intensity of crosstalk interaction is a function of the vertex degrees.
In this case, we denote the crosstalk affinity of vertex $u$ with degree $d_{u}$ and vertex $v$ with degree $d_{v}$ as 
\begin{equation}
w^{xt}_{u,v}=g(d_{u},d_{v})
\end{equation}
 which represents the intensity of the crosstalk interaction between $u$ and $v$.
The function $g$ is assumed to be a strictly increasing function in both its arguments, reflecting our assumption that the more neighbors a vertex has, the more likely it is to interact with vertices it is not connected with.
We also assume that $g$ is symmetric, 
\begin{equation}
g(d_{u},d_{v})=g(d_{v},d_{u}).
\end{equation}
The affinity $f=f(d_{u})$ of a vertex $u$ is similarly defined as its tendency to interact with other vertices of the network with which it is not connected, and consequently not supposed to interact.
The following sections present two scenarios, the first defines $g$ as the sum of the affinity functions of the two vertices, and the other defines $g$ as the product of their affinities.

%%%%%%%%%%%%%%%%%%%%%%%%%%%%%%%%%%%%%%%%%%%%%%%%
%%%%%%%%%%%%%%%%%%%%%%%%%%%%%%%%%%%%%%%%%%%%%%%%
%%%%%%%%%%%%%%%%%%%%%%%%%%%%%%%%%%%%%%%%%%%%%%%%
\section{Pairwise Crosstalk Interactions as the Sum of the Individual Affinities}

If $g(x,y)$ is an additive function, we can write it as
\begin{equation}
g(x,y)=f(x) +f(y)
\end{equation}
where $f:{\Bbb R}^{+} \to {\Bbb R}^{+}$ is defined as the affinity of each individual vertex, and is a strictly increasing function of its degree.
We will refer to this type of crosstalk interactions in the network as \textit{additive crosstalk}.
The overall sensitivity to crosstalk will then be equal to
\begin{align}
K &=\sum _{(u,v) \notin \mathcal{E}(\mathcal{G}) } g(d_{u},d_{v}) \notag \\
  & =\sum _{(u,v) \notin \mathcal{E}(\mathcal{G}) } \left( f(d_{u})+f(d_{v}) \right).
\label{DefinitionAdditiveEquation}
\end{align}

\subsection{General Form of Graphs with Minimum Sensitivity to Crosstalk}

We will first present some lemmas that will help us formulate the main theorem of this section.
\begin{lemma}
Assume that the function $h:{\Bbb R}^{+} \to {\Bbb R}^{+}$ can be written as the product of two functions $f$ and $g$, with $f:{\Bbb R}^{+} \to {\Bbb R}^{+}$ being strictly increasing, differentiable and concave function, and $g:{\Bbb R}^{+} \to {\Bbb R}^{+}$ is strictly decreasing, differentiable and concave.
Then $h$ is strictly concave.
\label{ConcaveProductLemma}
\end{lemma}
\begin{proof}
The second derivative of $h$ is
\begin{equation}
h''(x) =f''(x)g(x)+2f'(x)g'(x)+f(x)g''(x).
\end{equation}
All three products of the right hand side are negative, since $f(x)>0$, $f'(x)> 0$, $f''(x) \leq 0$, $g(x)>0$, $g'(x) < 0$ and $g''(x) \leq 0$, therefore $h''(x)<0$.
\end{proof}

\begin{lemma}
Assume we have a differentiable function $f:{\Bbb R}^{+} \to {\Bbb R}^{+}$, with $f'(x)<0$ and $f''(x) \leq 0$ for every $x \in \mathbb{R}^{+}$.
Further assume that for $a,b,c,d \in \mathbb{R}^{+}$, we have  $0\leq a< c\leq d < b<\infty$, and for which $a+b=c+d$.
Then $f(a)+f(b)\leq f(c)+f(d)$.
\label{ConcaveLemma}
\end{lemma}
\begin{proof}
According to the mean value theorem, there exist numbers $a \leq \xi _{1} \leq c$ and $d \leq \xi _{2} \leq b$ such that 
\begin{equation}
f^{'}(\xi _{1})=\frac{f(c)-f(a)}{c-a}
\end{equation}
and 
\begin{equation}
f^{'}(\xi _{2})=\frac{f(b)-f(d)}{b-d}.
\end{equation}
Since $f$ is concave,
\begin{equation}
f'(\xi_{2}) \leq f'(\xi_{1}).
\end{equation}
Replacing these derivatives from the equations above, we get
\begin{equation}
\frac{f(b)-f(d)}{b-d} \leq \frac{f(c)-f(a)}{c-a} .
\end{equation}
Since $a+b=c+d$, the denominators are equal by assumption, 
\begin{equation}
c-a=b-d
\end{equation}
and as a consequence
\begin{equation}
f(b)-f(d) \leq f(c)-f(a).
\end{equation}
Rearranging the terms, the result follows:
\begin{equation}
f(a)+f(b)  \leq f(c)+f(d).
\end{equation}
\end{proof}

\begin{definition}
The total affinity $h_{u}$ of a vertex $u$ is equal to the sum of the affinities of $u$ to all the vertices of the network that it is \textit{not} connected with. 
\end{definition}

\begin{theorem}
Assume that the function $f:{\Bbb R}^{+} \to {\Bbb R}^{+}$ that defines the affinity of each vertex of the network is a strictly increasing  concave function of the vertex degree.
Then the total affinity of every vertex $u$ is equal to 
\begin{equation}
h_{u}=(N-1-d_{u})f(d_{u})
\end{equation}
and the degree distribution of the network with the minimum additive sensitivity to crosstalk is inhomogeneous.
\end{theorem}
\begin{proof}
As shown in equation (\ref{DefinitionAdditiveEquation}), the crosstalk intensity between any pair of vertices is only a function of their degrees.
Furthermore, the affinity of each vertex is only a function of its degree,  and independent of the vertices it interacts with.
As a result,
\begin{equation}
h_{u}=\sum _{v \notin \mathcal{N}_{u}}f(d_{u})  \implies h_{u}= (N-1-d_{u})f(d_{u}).
\end{equation}
Since $h_{u}$ is also a function of the vertex degree, we will denote the total affinity as $h(d_{u})$.
From Lemma \ref{ConcaveProductLemma} , it is easy to see that $h$ is a positive, strictly increasing and concave function, since $N-1-d_{u}$ is a decreasing positive linear function for $1\leq d_{u} \leq N-1$ and $f$ is assumed to be positive, increasing and concave.

The network which is most robust to crosstalk will be the one with degree distribution $\mathbf{d}^{*}=[d_{1}, d_{2}, \hdots , d_{N}]$   such that 
\begin{equation}
\mathbf{d} ^{*} =\argmin _{\mathbf{d} \in \mathcal{P}}\left\{  \sum _{(u,v) \notin \mathcal{E}(\mathcal{G}) }  \left( f(d_{u})+f(d_{v}) \right)   \right\} 
\end{equation}
under the constraints
\begin{equation}
1 \leq d_{u} \leq N-1 \quad  \forall u \in \mathcal{V}(\mathcal{G})
\label{PlausibleDegreesInequality}
\end{equation}
\begin{equation}
\sum _{u=1}^{N}d_{u}=2m \geq 2(N-1)
\end{equation}
where $\mathcal{P}$ is the space of graphic degree sequences for connected graphs.
Rewriting equation (\ref{DefinitionAdditiveEquation}), we get:
\begin{align}
K =&\sum _{(u,v) \notin \mathcal{E}(\mathcal{G}) }  \left( f(d_{u})+f(d_{v}) \right) \notag \\
=&\sum _{u \in \mathcal{V}(\mathcal{G}) }  (N-1-d_{u}) f(d_{u}) \label{VertexOnlyEquation} \\
=&\sum _{u \in \mathcal{V}(\mathcal{G}) }  h(d_{u}) \notag.
\end{align}
The function $K$ is a sum of concave functions, and thus concave.
According to Lemma \ref{ConcaveLemma}, the optimal $\mathbf{d}^{*}$ will be achieved for $d_{u}$'s at the boundaries of the plausible values in the inequality (\ref{PlausibleDegreesInequality}).
\end{proof}

\begin{corollary}
When the crosstalk affinity between two vertices is an additive function of their degrees, the overall crosstalk sensitivity only depends on the vertex degrees, and does not depend on the structure of the network. 
In other words, how robust the network is to crosstalk only depends on its degree distribution.
\end{corollary}
\begin{proof}
From equation (\ref{VertexOnlyEquation}), the overall sensitivity of the network only depends on the degree distribution, and is independent of its structure.
\end{proof}

\begin{corollary}
The tree graph  that is most robust to additive crosstalk is the star graph.
\end{corollary}
\begin{proof}
The star network has one vertex with full degree and this vertex has no crosstalk interactions.
The $N-1$ peripheral vertices have degree $1$, and may spuriously interact with each other.
Hence,
\begin{align}
K_{star} =&\sum _{u \in \mathcal{V}(\mathcal{G}) }  (N-1-d_{u}) f(d_{u}) \notag \\
=&(N-1)(N-2) f(1)+ (N-1)\cdot 0\\
=&2 {N-1 \choose 2} f(1). \notag
\end{align}

Every other graph will have vertices with degrees $1\leq d \leq N-2$ with at least one vertex with degree $d \geq 2$.
The total number of crosstalk interactions is equal to twice the number of edges that are \textit{not} present in the network, namely $2{N-1 \choose 2}$.
The sum of all the crosstalk interactions $K_{alt}$ will then be 
\begin{align}
K_{alt} &=\sum _{u \in \mathcal{V}(\mathcal{G}) }  (N-1-d_{u}) f(d_{u}) \notag \\
&>f(1) \sum _{u \in \mathcal{V}(\mathcal{G}) }  (N-1-d_{u}) \notag \\
&=f(1) (N(N-1)-2(N-1)) \\
&=2 {N-1 \choose 2} f(1) \notag \\
&=K_{star}. \notag
\end{align}
\end{proof}

\begin{corollary}
Suppose that a vertex $u$ is connected to vertex $v$ with degree $d_{v}$, and is not connected to vertex $w$ with degree $d_{w}$.
If $d_{v} \leq d_{w}$, and we can rewire the edge $(u,v)$ to $(u,w)$ while keeping the graph connected, then the sensitivity of the new network to additive crosstalk is smaller.
\label{IncreasingDegreesCorollary}
\end{corollary}
\begin{proof}
The only vertices whose affinities change are $v$ and $w$, since the degree of $u$, along with any other vertex in the network, stays the same.
Thus,
\begin{align}
\Delta K =&K_{new}-K_{old} \notag\\
=&(N-d_{v})f(d_{v}-1)+(N-2-d_{w})f(d_{w}-1) \notag \\
&\quad -(N-1-d_{v})f(d_{v})-(N-1-d_{w})f(d_{w})\\
=& -\left[ h(d_{v})-h(d_{v}-1) \right]+\left[ h(d_{w}+1)-h(d_{w}) \right] \notag \\
<&0. \notag
\end{align}
The last inequality holds because $h$ is concave (Lemma \ref{ConcaveProductLemma}) and $d_{v}<d_{w}$.
\end{proof}

\begin{definition}
We define $\mathcal{N}_{x-y}$ as the set of all the neighbors of vertex $x$ except vertex $y$, such that 
\begin{equation}
\mathcal{N}_{x-y}= \left\{ \begin{array}{ll}
\mathcal{N}_{x} & \textrm{if $y \notin \mathcal{N}_{x}$}\\
\mathcal{N}_{x} - \{ y\} & \textrm{if $y \in  \mathcal{N}_{x}$}.\\
\end{array} \right.
\end{equation}
\end{definition}

\begin{corollary}
In a graph with the minimum crosstalk sensitivity, for every pair of vertices $a$ and $b$, we have
\begin{equation}
d_{a}\leq d_{b} \iff   \mathcal{N}_{a-b} \subseteq \mathcal{N}_{b-a}.
\label{DegreeSubsetCondition}
\end{equation}
\end{corollary}
\begin{proof}
Counting the elements of the sets $\mathcal{N}_{a-b}$ and $\mathcal{N}_{b-a}$, 
\begin{equation}
\mathcal{N}_{a-b} \subseteq \mathcal{N}_{b-a} \implies |\mathcal{N}_{a-b}|\leq |\mathcal{N}_{b-a}|.
\end{equation}
In addition,
\begin{equation}
|\mathcal{N}_{a}|=|\mathcal{N}_{a-b}|+\delta((a,b)\in \mathcal{E}), \quad |\mathcal{N}_{b}|=|\mathcal{N}_{b-a}|+\delta((a,b)\in \mathcal{E})
\end{equation}
where $\delta$ is the Kronecker delta function, defined as
\begin{equation}
\delta((u,v) \in \mathcal{E}) = \left\{ \begin{array}{ll}
1 & \textrm{if $u$ and $v$ are connected}\\
0 & \textrm{otherwise.}
\end{array} \right.
\end{equation}
As a result,  
\begin{equation}
d_{a}=|\mathcal{N}_{a} |\leq |\mathcal{N}_{b}| =d_{b}.
\end{equation}
The necessity of condition (\ref{DegreeSubsetCondition}) is proved, and we will next show the sufficiency of the statement by assuming otherwise and arriving at a contradiction.

Suppose that $d_{a} \leq d_{b}$, but $\mathcal{N}_{a-b} \nsubseteq \mathcal{N}_{b-a}$.
This means that there is at least one vertex that is connected to $a$, but is not connected to $b$.
Let $M \in \mathcal{V}(\mathcal{G})$ be a vertex with degree $d_{M}$ such that
\begin{equation}
d_{M}\geq d_{u} \quad \forall u \in \mathcal{N}_{a}\cap\mathcal{N}_{b}^{c}.
\end{equation}
Since the graph $\mathcal{G}$ has the minimum sensitivity to crosstalk, by Corollary \ref{IncreasingDegreesCorollary}, this means that $a$ is connected to all vertices that have degree larger than $M$.
If it did not, we could rewire the edge $(a,M)$ to connect $a$ with a vertex with a larger degree.
The same applies for vertex $b$, and since $d_{a}\leq d_{b}$, the only way that vertex $b$ would not be connected to $M$, is when it is connected to some other vertex $P$ with degree equal to the degree of $M$, 
\begin{equation}
d_{M}=d_{P}=\theta .
\end{equation}
Even in that case, if we rewire the edge that connects $b$ with $P$, so that after rewiring it connects $b$ with $M$,  the difference of overall crosstalk sensitivity for the graph $\mathcal{G}$ will be negative
\begin{align} 
\Delta K &= K_{new}-K_{old} \notag \\
& =h(\theta +1)+h(\theta-1)-2h(\theta) \\
& <0 \notag
\end{align}
because the total crosstalk affinity function $h$ is strictly concave.
This means that $\mathcal{G}$ was not a graph with minimum crosstalk sensitivity, which is a contradiction.

Finally note that the last derivation confirms that in the graph with minimum crosstalk sensitivity, 
\begin{equation}
d_{a}=d_{b} \iff  \mathcal{N}_{a-b}=\mathcal{N}_{b-a}.
\label{EqualDegreesSameNeighbors}
\end{equation}
\end{proof}

\subsection{Rewiring Algorithm}

In this section, we describe an algorithm that finds the structure of the network with the minimum sensitivity to additive crosstalk.
We will first explain how it works, and then prove that it always returns the optimal network, regardless of the input graph.
The algorithm consists of two steps: The first step (single rewirings) ensures that the condition of equation (\ref{DegreeSubsetCondition}) holds.
In every iteration, it checks if each vertex could be connected to another neighbor with larger degree, therefore decreasing the overall crosstalk sensitivity of the graph, according to  Corollary \ref{IncreasingDegreesCorollary}.
The second step ensures that the required changes that are not possible with single-vertex rewirings actually take place.

\begin{figure}[htbp]
\centering
\psscalebox{0.3}{
\begin{pspicture}(-6,-6)(20,6)
{
\cnodeput[fillstyle=solid,fillcolor=red](2.5,4.33){A}{\strut}
\cnodeput[fillstyle=solid,fillcolor=green](-2.5,4.33){B}{\strut}
\cnodeput[fillstyle=solid,fillcolor=red](-5,0){C}{\strut}
\cnodeput(-2.5,-4.33){D}{\strut}
\cnodeput[fillstyle=solid,fillcolor=green](2.5,-4.33){E}{\strut}
\cnodeput(5,0){F}{\strut}
\cnodeput[fillstyle=solid,fillcolor=red](16.5,4.33){L}{\strut}
\cnodeput[fillstyle=solid,fillcolor=green](11.5,4.33){H}{\strut}
\cnodeput[fillstyle=solid,fillcolor=red](9,0){I}{\strut}
\cnodeput(11.5,-4.33){J}{\strut}
\cnodeput[fillstyle=solid,fillcolor=green](16.5,-4.33){K}{\strut}
\cnodeput(19,0){G}{\strut}
}
\ncline{-}{A}{C}
\ncline{-}{A}{F}
\ncline{-}{C}{B}
\ncline{-}{C}{D}
\ncline{-}{C}{E}
\ncline{-}{C}{F}
\ncline{-}{F}{B}
\ncline{-}{F}{D}
\ncline{-}{F}{E}
\ncline{-}{D}{E}
\ncline{-}{B}{D}
\rput(7,0){\Huge $\Rightarrow$}
\ncline{-}{G}{H}
\ncline{-}{G}{I}
\ncline{-}{G}{J}
\ncline{-}{G}{K}
\ncline{-}{G}{L}
\ncline{-}{H}{I}
\ncline{-}{H}{J}
\ncline{-}{H}{K}
\ncline{-}{I}{J}
\ncline{-}{I}{K}
\ncline{-}{J}{K}
\end{pspicture}
}
\caption{An example of a double rewiring needed for the graph to minimize crosstalk sensitivity. If $h_{1}+3h_{4}<h_{2}+2h_{3}+h_{5}$ (for example when $h(d)=(N-1-d)\sqrt{d}$), there is no edge that can be rewired keeping one of its endpoints, and that would decrease the overall sensitivity.
A double rewiring is needed, removing the edge between the red vertices and connecting the green ones with it.}
\label{DoubleRewiringMotivation}
\end{figure}
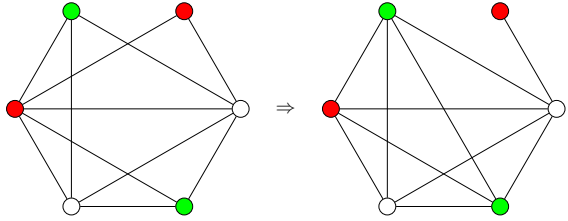

After the first step is finished (there are no more possible single-vertex rewirings),   the degrees of all vertices satisfy equation (\ref{DegreeSubsetCondition}).
The resulting graph is not necessarily one with the minimum possible sensitivity,  because we may have a situation like the one depicted in Figure \ref{DoubleRewiringMotivation}, where we need to rewire one edge such that we change both vertices it is adjacent to.
In this case, we cannot perform a single rewiring, and we need to delete an edge from a pair of vertices, and add it to another.
There are only two possible arrangements of the edges among four vertices that may need a double rewiring transformation (shown in Figure \ref{AllPossible4VertexCases}), and the condition that the degrees of these vertices need to satisfy is shown in Table \ref{DegreeOrderDoubleRewiringsTable}.
\begin{figure}[htbp]
\subfigure[]{
\centering
\psscalebox{0.4}{
\begin{pspicture}(-1,-1)(5.5,5)
{
\cnodeput(0,4){A}{\strut\boldmath$A$}
\cnodeput(4,4){B}{\strut\boldmath$B$}
\cnodeput(0,0){C}{\strut\boldmath$C$}
\cnodeput(4,0){D}{\strut\boldmath$D$}
}
\ncline{-}{A}{B}
\end{pspicture}
}
}
\subfigure[]{
\psscalebox{0.4}{
\begin{pspicture}(-1,-1)(5.5,5)
{
\cnodeput(0,4){A}{\strut\boldmath$A$}
\cnodeput(4,4){B}{\strut\boldmath$B$}
\cnodeput(0,0){C}{\strut\boldmath$C$}
\cnodeput(4,0){D}{\strut\boldmath$D$}
}
\ncline{-}{A}{B}
\ncline{-}{A}{C}
\end{pspicture}
}
}
\subfigure[]{
\psscalebox{0.4}{
\begin{pspicture}(-1,-1)(5.5,5)
{
\cnodeput(0,4){A}{\strut\boldmath$A$}
\cnodeput(4,4){B}{\strut\boldmath$B$}
\cnodeput(0,0){C}{\strut\boldmath$C$}
\cnodeput(4,0){D}{\strut\boldmath$D$}
}
\ncline{-}{A}{B}
\ncline{-}{A}{C}
\ncline{-}{A}{D}
\end{pspicture}
}
}
\subfigure[]{
\psscalebox{0.4}{
\begin{pspicture}(-1,-1)(5.5,5)
{
\cnodeput(0,4){A}{\strut\boldmath$A$}
\cnodeput(4,4){B}{\strut\boldmath$B$}
\cnodeput(0,0){C}{\strut\boldmath$C$}
\cnodeput(4,0){D}{\strut\boldmath$D$}
}
\ncline{-}{A}{B}
\ncline{-}{A}{C}
\ncline{-}{B}{C}
\end{pspicture}
}
}
\subfigure[]{
\psscalebox{0.4}{
\begin{pspicture}(-1,-1)(5.6,5)
{
\cnodeput(0,4){A}{\strut\boldmath$A$}
\cnodeput(4,4){B}{\strut\boldmath$B$}
\cnodeput(0,0){C}{\strut\boldmath$C$}
\cnodeput(4,0){D}{\strut\boldmath$D$}
}
\ncline{-}{A}{B}
\ncline{-}{A}{C}
\ncline{-}{B}{D}
\end{pspicture}
}
}
\subfigure[]{
\psscalebox{0.4}{
\begin{pspicture}(-2,-1)(6.5,5)
{
\cnodeput(0,4){A}{\strut\boldmath$A$}
\cnodeput(4,4){B}{\strut\boldmath$B$}
\cnodeput(0,0){C}{\strut\boldmath$C$}
\cnodeput(4,0){D}{\strut\boldmath$D$}
}
\ncline{-}{A}{B}
\ncline{-}{A}{C}
\ncline{-}{A}{D}
\ncline{-}{B}{C}
\end{pspicture}
}
}
\subfigure[]{
\psscalebox{0.4}{
\begin{pspicture}(-1.2,-1)(5.3,5)
{
\cnodeput(0,4){A}{\strut\boldmath$A$}
\cnodeput(4,4){B}{\strut\boldmath$B$}
\cnodeput(0,0){C}{\strut\boldmath$C$}
\cnodeput(4,0){D}{\strut\boldmath$D$}
}
\ncline{-}{A}{B}
\ncline{-}{A}{C}
\ncline{-}{A}{D}
\ncline{-}{B}{C}
\ncline{-}{B}{D}
\end{pspicture}
}
}
\caption{All nonisomorphic subgraphs of $4$ vertices where $A - B$ and $C\notslash D$. See Table \ref{DegreeOrderDoubleRewiringsTable} for all the allowed degree correlations among the degrees of the four vertices that satisfy equation (\ref{DegreeSubsetCondition}). Subgraphs $(c)$ and $(d)$ may need a rewiring of an edge between $A$ and $B$ to connect vertices $C$ and $D$. Note that if we apply the double-vertex rewiring, we transform subgraph $(c)$ to subgraph $(d)$ and vice versa.}
\label{AllPossible4VertexCases}
\end{figure}
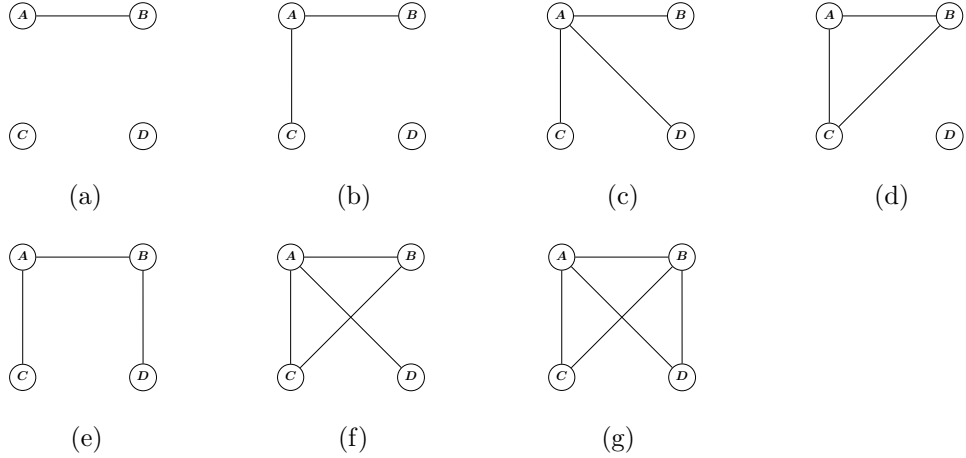
\begin{table}[htbp]
\begin{center}
\begin{tabular}{|c|c|c|}
\hline
Graph & Possible Degree Order & $\Delta K$ \\
\hline
\hline
a & $d_{D}\leq d_{C} \leq d_{B} \leq d_{A}$ & $>0$ \\
\hline
b & $d_{D}\leq d_{C} \leq d_{B} \leq d_{A}$ & $>0$ \\
   & $d_{D}\leq d_{B} \leq d_{C} \leq d_{A}$ & $>0$ \\
\hline
   & $d_{D}\leq d_{C} \leq d_{B} \leq d_{A}$ & $>0$ \\
c & $d_{D}\leq d_{B} \leq d_{C} \leq d_{A}$ & $>0$ \\
   & $d_{B}\leq d_{D} \leq d_{C} \leq d_{A}$ & \textbf{\red{?}} \\
\hline
   & $d_{D}\leq d_{C} \leq d_{B} \leq d_{A}$ & $>0$ \\
d & $d_{D}\leq d_{B} \leq d_{C} \leq d_{A}$ & $>0$ \\
   & $d_{D}\leq d_{B} \leq d_{A} \leq d_{C}$ & \textbf{\red{?}} \\
\hline
e & $d_{D}\leq d_{C} \leq d_{B} = d_{A}$ & $>0$ \\
\hline
f & $d_{D}\leq d_{C} \leq d_{B} \leq d_{A}$ & $>0$ \\
   & $d_{D}\leq d_{B} \leq d_{C} \leq d_{C}$ & $>0$ \\
\hline
g & $d_{D}\leq d_{C} \leq d_{B} \leq d_{A}$ & $>0$ \\
\hline
\end{tabular}
\end{center}
\caption{All the possible degree orderings of the degrees of the four vertices shown in  Figure \ref{AllPossible4VertexCases}, with the difference in the overall crosstalk sensitivity if we rewire the edge between $A$ and $B$ to connect vertices $C$ and $D$, and assuming that all possible single rewirings that decrease crosstalk have already taken place. In all but two cases, the overall sensitivity to crosstalk increases. Cases $(c)$ and $(d)$ may have a positive or negative difference (red question marks), depending on the crosstalk function $h$,and only if the degrees of the vertices $A,B,C,D$ are ordered as shown. 
Note that case $(e)$ is  not possible at any time, but is included here for completeness.
}
\label{DegreeOrderDoubleRewiringsTable}
\end{table}

\begin{algorithm}
\caption{Find the Graph with Minimum Crosstalk Sensitivity $\mathcal{G}$}
\label{MinAdditiveXtalkAlgorithm}
\begin{algorithmic}
\REQUIRE An arbitrary connected graph $\mathcal{G}$ of order $N$ and size $m$
\ENSURE The graph with the minimum sensitivity to crosstalk
\STATE
\STATE changeflag$\leftarrow 1$
\WHILE{changeflag==$1$}
\STATE changeflag$\leftarrow 0$
\FOR{$u \in \mathcal{V}(\mathcal{G})$}
\IF{$ \displaystyle \min _{v \in \mathcal{N}_{u} } \{ d_{v}\}  \leq  \max _{w \in \mathcal{N}^{c}_{u} } \{ d_{w}\}  $}
\STATE $\displaystyle x=\argmin _{v \in \mathcal{N}_{u} } \{ d_{v}\}$; $\displaystyle y=\argmin _{w \in \mathcal{N}_{u}^{c} } \{ d_{w}\}$
\STATE Remove the edge $(u,x)$;Add the edge $(u,y)$; changeflag$\leftarrow 1$
\IF{$\mathcal{G}$ becomes disconnected}
\STATE Revert the changes
\ENDIF
\ENDIF
\ENDFOR
\ENDWHILE
\STATE 
\STATE changeflag$\leftarrow 1$
\WHILE{changeflag==$1$}
\STATE changeflag$\leftarrow 0$
\FOR{$A,B,C,D \in \mathcal{V}(\mathcal{G})$}
\IF{$(A - B $  \AND  $A - C $ \AND  $A - D $ \AND  $B \notslash C $ \AND  $B \notslash D $ \AND  $C \notslash D) $ \OR $(A - B $  \AND  $A - C $ \AND  $B - C $ \AND  $A \notslash D $ \AND  $B \notslash D $ \AND  $C \notslash D $) }
\STATE $\Delta G_{A}=g(d_{A})-g(d_{A}-1)$;$\Delta G_{B}=g(d_{B})-g(d_{B}-1)$
\STATE $\Delta G_{C}=g(d_{C}+1)-g(d_{C})$; $\Delta G_{D}=g(d_{D}+1)-g(d_{D})$
\STATE $\Delta K_{A,B,C,D}=-\Delta G_{A}-\Delta G_{B}+\Delta G_{C}+\Delta G_{D}   $
\ENDIF
\STATE Find the vertex set $(A,B,C,D)$ with the smallest $\Delta K$
\IF{$\min\Delta K<0$ \AND rewiring the edge $A-B$ to $C-D$ keeps $\mathcal{G}$ connected}
\STATE Remove the edge $(A,B)$; Add the edge $(C,D)$; changeflag$\leftarrow 1$
\ENDIF
\ENDFOR
\ENDWHILE
\end{algorithmic}
\end{algorithm}

\begin{theorem}
Provided with an arbitrary graph $\mathcal{G}_{N,m}$, Algorithm \ref{MinAdditiveXtalkAlgorithm} returns the network $\mathcal{G}^{*}_{N,m}$ with the minimum sensitivity to additive crosstalk.
\end{theorem}
\begin{proof}
We need to prove two statements, the first is that the double rewirings do not affect the condition of equation (\ref{DegreeSubsetCondition}), and the second is that the algorithm cannot be trapped in a local minimum.
Suppose that we perform a double rewiring (removing the edge $(A,B)$ and adding the edge (C,D)) reducing the overall sensitivity to crosstalk by $\Delta K$, and condition (\ref{DegreeSubsetCondition}) is violated.
This would only be possible if there exists a vertex $M$ with degree $d_{M}$ such that 
\begin{equation}
d_{M}=d_{B}-1 \quad \textrm{and} \quad \mathcal{N}_{B-A} \neq \mathcal{N}_{M},
\end{equation}
so that after removing the edge $(A,B)$, vertices $B$ and $M$ have the same degrees, and condition (\ref{EqualDegreesSameNeighbors}) does not hold.
The last equation also implies that $(A,M) \in \mathcal{E}$, because otherwise $\mathcal{N}_{B}$ and $\mathcal{N}_{M}$ would be identical after removing the edge $(A,B)$.
But in that case, 
\begin{equation}
\Delta K'=-\Delta h_{A}-\Delta h_{M}+\Delta h_{C}+\Delta h_{D}<-\Delta h_{A}-\Delta h_{B}+\Delta h_{C}+\Delta h_{D}= \Delta K
\end{equation}
because $h$ is strictly concave, thus $\Delta h_{M}>\Delta h_{B}$.

We will now prove that the algorithm cannot be trapped in a local minimum by showing that all the individual steps needed in order to reach the graph with the minimum  overall crosstalk sensitivity can be reordered arbitrarily.
This condition is sufficient in order to prove that when there are no such rewirings left, we have reached the optimal graph (Figure \ref{ReorderRewirings}).

\begin{figure}[htbp]
\begin{center}
\includegraphics[scale=0.4]{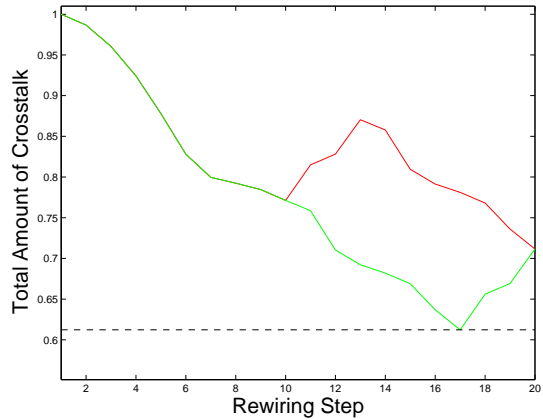} 
\caption{
Algorithm \ref{MinAdditiveXtalkAlgorithm} cannot be trapped in a local minimum. After decreasing the overall crosstalk sensitivity (blue line), suppose it stops at a local minimum. Since we can reorder the steps, and each step decreases the sensitivity at least as much as when taken individually (equation (\ref{SubsequentSmallerEquation})), we may perform all the decreasing transformations first and then stop, ending up with a network with smaller overall sensitivity to crosstalk.
}
\label{ReorderRewirings}
\end{center}
\end{figure}

The single-vertex rewirings can be performed in an arbitrary order, since the only way there are no more such rewirings possible is when the network satisfies equation (\ref{DegreeSubsetCondition}).
So it suffices to prove that the double rewirings can be rearranged.
Suppose that we have two or more such rewirings needed.
For each pair of vertices that will accept one extra edge, they can have up to one vertex in common, whose degree will increase by at least $2$.
If they have no vertex in common, the double rewirings are completely independent of each other, so the statement holds.
They cannot have the same vertices, because these two vertices will already be connected by an edge after one of the rewirings.
Now assume that we can perform two double rewirings.
The first is moving the edge $(A,B)$ to $(C,D)$, changing crosstalk by $\Delta K_{1}$, and the second is moving the edge $(E,F)$ to $(C,J)$, changing crosstalk by $\Delta K_{2}$, based on the current degrees of vertices $A,B,C,D,E,F$ and $J$. 
If $C$ is common in both rewirings, then if we perform the first rewiring and then perform the second rewiring (now changing the overall crosstalk sensitivity by $\Delta K_{2}'$), we will show that
\begin{equation}
\Delta K _{1}<0, \Delta K_{2}<0 \implies \Delta K_{2}' \leq \Delta K_{2}.
\end{equation}
This means that every rewiring makes the subsequent ones decrease the crosstalk even more than their individual contributions, and there is no way that a rewiring that increases crosstalk could enable future rewirings to decrease the overall crosstalk more than they would have otherwise.
Denote
\begin{equation}
\Delta K_{1}=-\Delta h_{A}-\Delta h_{B}+\Delta h_{C}+\Delta h_{D}<0
\end{equation}
where 
\begin{equation}
\Delta h_{A}=h(d_{A})-h(d_{A}-1)>0,\qquad \Delta h_{C}=h(d_{C}+1)-h(d_{C})>0
\end{equation}
and similarly for the differences of the crosstalk affinities of the other vertices.
By the same token,
\begin{equation}
\Delta K_{2}=-\Delta h_{E}-\Delta h_{F}+\Delta h_{C}+\Delta h_{J}<0.
\end{equation}
If we perform the first double rewiring ($(A,B)$ is rewired to $(C,D)$), the new degree of vertex $C$ will increase by $1$, and as a result
\begin{equation}
\Delta K_{2}'=-\Delta h_{E}-\Delta h_{F}+\Delta h'_{C}+\Delta h_{J}
\end{equation}
where
\begin{equation}
\Delta h'_{C}=h(d_{C}+2)-h(d_{C}+1)<\Delta h_{C}
\end{equation}
since $h$ is strictly concave.
It follows that in every case we have
\begin{equation}
\Delta K_{2}' \leq \Delta K_{2}.
\label{SubsequentSmallerEquation}
\end{equation}
The same is true when we have more than one double rewiring that decreases the overall crosstalk sensitivity and removes two or more edges from a vertex, because in that case, we can similarly show that the difference of the crosstalk affinity of this vertex increases, and since it is negated, the decrease in crosstalk of the subsequent rewirings will be greater.
\end{proof}

\begin{theorem}
We can find the graph with the lowest sensitivity to crosstalk in $\mathcal{O}(N^{4})$ time.
\end{theorem}
\begin{proof}
The first part of the algorithm checks if each edge should be rewired, which takes 
$\mathcal{O}(N^{2})$ steps, since we need to check all the vertex pairs, possibly more than once.
The second part of the algorithm needs to check groups of $4$ vertices at a time, which will take $\mathcal{O}(N^{4})$ steps.
So the overall complexity of Algorithm \ref{MinAdditiveXtalkAlgorithm} is $\mathcal{O}(N^{4})$. 
\end{proof}

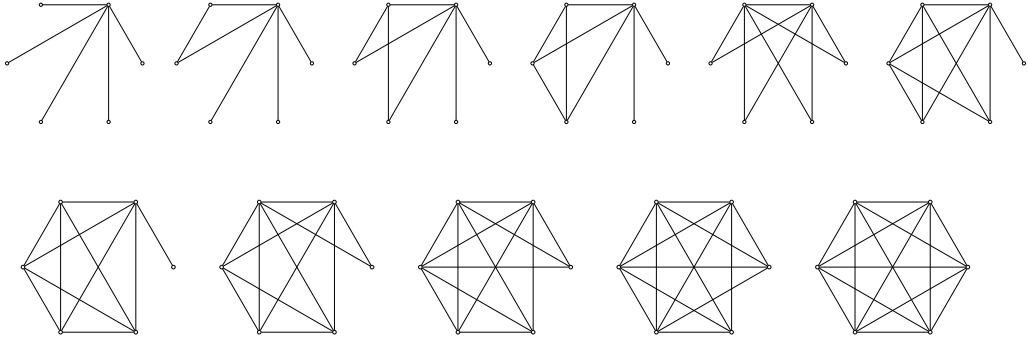
\begin{figure}[htbp]
\subfigure{%5
\psscalebox{0.18}{
\begin{pspicture}(-5,-5)(5,6)
{
\cnodeput(2.5,4.33){A}{}
\cnodeput(-2.5,4.33){B}{}
\cnodeput(-5,0){C}{}
\cnodeput(-2.5,-4.33){D}{}
\cnodeput(2.5,-4.33){E}{}
\cnodeput(5,0){F}{}
}
\ncline{-}{A}{B}
\ncline{-}{A}{C}
\ncline{-}{A}{D}
\ncline{-}{A}{E}
\ncline{-}{A}{F}
\end{pspicture}
}
}
\subfigure{%6
\psscalebox{0.18}{
\begin{pspicture}(-5,-5)(5,8)
{
\cnodeput(2.5,4.33){A}{}
\cnodeput(-2.5,4.33){B}{}
\cnodeput(-5,0){C}{}
\cnodeput(-2.5,-4.33){D}{}
\cnodeput(2.5,-4.33){E}{}
\cnodeput(5,0){F}{}
}
\ncline{-}{A}{B}
\ncline{-}{A}{C}
\ncline{-}{A}{D}
\ncline{-}{A}{E}
\ncline{-}{A}{F}
\ncline{-}{B}{C}
\end{pspicture}
}
}
\subfigure{
\psscalebox{0.18}{%7
\begin{pspicture}(-5,-5)(5,8)
{
\cnodeput(2.5,4.33){A}{}
\cnodeput(-2.5,4.33){B}{}
\cnodeput(-5,0){C}{}
\cnodeput(-2.5,-4.33){D}{}
\cnodeput(2.5,-4.33){E}{}
\cnodeput(5,0){F}{}
}
\ncline{-}{A}{B}
\ncline{-}{A}{C}
\ncline{-}{A}{D}
\ncline{-}{A}{E}
\ncline{-}{A}{F}
\ncline{-}{B}{C}
\ncline{-}{B}{D}
\end{pspicture}
}
}
\subfigure{
\psscalebox{0.18}{%8
\begin{pspicture}(-5,-5)(5,8)
{
\cnodeput(2.5,4.33){A}{}
\cnodeput(-2.5,4.33){B}{}
\cnodeput(-5,0){C}{}
\cnodeput(-2.5,-4.33){D}{}
\cnodeput(2.5,-4.33){E}{}
\cnodeput(5,0){F}{}
}
\ncline{-}{A}{B}
\ncline{-}{A}{C}
\ncline{-}{A}{D}
\ncline{-}{A}{E}
\ncline{-}{A}{F}
\ncline{-}{B}{C}
\ncline{-}{B}{D}
\ncline{-}{C}{D}
\end{pspicture}
}
}
\subfigure{
\psscalebox{0.18}{%9
\begin{pspicture}(-5,-5)(5,8)
{
\cnodeput(2.5,4.33){A}{}
\cnodeput(-2.5,4.33){B}{}
\cnodeput(-5,0){C}{}
\cnodeput(-2.5,-4.33){D}{}
\cnodeput(2.5,-4.33){E}{}
\cnodeput(5,0){F}{}
}
\ncline{-}{A}{B}
\ncline{-}{A}{C}
\ncline{-}{A}{D}
\ncline{-}{A}{E}
\ncline{-}{A}{F}
\ncline{-}{B}{C}
\ncline{-}{B}{D}
\ncline{-}{B}{E}
\ncline{-}{B}{F}
\end{pspicture}
}
}
\subfigure{
\psscalebox{0.18}{%10
\begin{pspicture}(-5,-5)(5,8)
{
\cnodeput(2.5,4.33){A}{}
\cnodeput(-2.5,4.33){B}{}
\cnodeput(-5,0){C}{}
\cnodeput(-2.5,-4.33){D}{}
\cnodeput(2.5,-4.33){E}{}
\cnodeput(5,0){F}{}
}
\ncline{-}{A}{B}
\ncline{-}{A}{C}
\ncline{-}{A}{D}
\ncline{-}{A}{E}
\ncline{-}{A}{F}
\ncline{-}{B}{C}
\ncline{-}{B}{D}
\ncline{-}{B}{E}
\ncline{-}{C}{D}
\ncline{-}{C}{E}
\end{pspicture}
}
}
\subfigure{
\psscalebox{0.2}{%11
\begin{pspicture}(-5.5,-5)(5.5,8)
{
\cnodeput(2.5,4.33){A}{}
\cnodeput(-2.5,4.33){B}{}
\cnodeput(-5,0){C}{}
\cnodeput(-2.5,-4.33){D}{}
\cnodeput(2.5,-4.33){E}{}
\cnodeput(5,0){F}{}
}
\ncline{-}{A}{B}
\ncline{-}{A}{C}
\ncline{-}{A}{D}
\ncline{-}{A}{E}
\ncline{-}{A}{F}
\ncline{-}{B}{C}
\ncline{-}{B}{D}
\ncline{-}{C}{D}
\ncline{-}{B}{E}
\ncline{-}{C}{E}
\ncline{-}{D}{E}
\end{pspicture}
}
}
\subfigure{
\psscalebox{0.2}{%12
\begin{pspicture}(-5.5,-5)(5.5,8)
{
\cnodeput(2.5,4.33){A}{}
\cnodeput(-2.5,4.33){B}{}
\cnodeput(-5,0){C}{}
\cnodeput(-2.5,-4.33){D}{}
\cnodeput(2.5,-4.33){E}{}
\cnodeput(5,0){F}{}
}
\ncline{-}{A}{B}
\ncline{-}{A}{C}
\ncline{-}{A}{D}
\ncline{-}{A}{E}
\ncline{-}{A}{F}
\ncline{-}{B}{C}
\ncline{-}{B}{D}
\ncline{-}{B}{E}
\ncline{-}{B}{F}
\ncline{-}{C}{D}
\ncline{-}{C}{E}
\ncline{-}{D}{E}
\end{pspicture}
}
}
\subfigure{
\psscalebox{0.2}{%13
\begin{pspicture}(-5.5,-5)(5.5,8)
{
\cnodeput(2.5,4.33){A}{}
\cnodeput(-2.5,4.33){B}{}
\cnodeput(-5,0){C}{}
\cnodeput(-2.5,-4.33){D}{}
\cnodeput(2.5,-4.33){E}{}
\cnodeput(5,0){F}{}
}
\ncline{-}{A}{B}
\ncline{-}{A}{C}
\ncline{-}{A}{D}
\ncline{-}{A}{E}
\ncline{-}{A}{F}
\ncline{-}{B}{C}
\ncline{-}{B}{D}
\ncline{-}{C}{D}
\ncline{-}{B}{E}
\ncline{-}{C}{E}
\ncline{-}{D}{E}
\ncline{-}{B}{F}
\ncline{-}{C}{F}
\end{pspicture}
}
}
\subfigure{
\psscalebox{0.2}{%14
\begin{pspicture}(-5.5,-5)(5.5,8)
{
\cnodeput(2.5,4.33){A}{}
\cnodeput(-2.5,4.33){B}{}
\cnodeput(-5,0){C}{}
\cnodeput(-2.5,-4.33){D}{}
\cnodeput(2.5,-4.33){E}{}
\cnodeput(5,0){F}{}
}
\ncline{-}{A}{B}
\ncline{-}{A}{C}
\ncline{-}{A}{D}
\ncline{-}{A}{E}
\ncline{-}{A}{F}
\ncline{-}{B}{C}
\ncline{-}{B}{D}
\ncline{-}{C}{D}
\ncline{-}{B}{E}
\ncline{-}{C}{E}
\ncline{-}{D}{E}
\ncline{-}{B}{F}
\ncline{-}{C}{F}
\ncline{-}{D}{F}
\end{pspicture}
}
}
\subfigure{
\psscalebox{0.2}{%15
\begin{pspicture}(-5.5,-5)(5.5,8)
{
\cnodeput(2.5,4.33){A}{}
\cnodeput(-2.5,4.33){B}{}
\cnodeput(-5,0){C}{}
\cnodeput(-2.5,-4.33){D}{}
\cnodeput(2.5,-4.33){E}{}
\cnodeput(5,0){F}{}
}
\ncline{-}{A}{B}
\ncline{-}{A}{C}
\ncline{-}{A}{D}
\ncline{-}{A}{E}
\ncline{-}{A}{F}
\ncline{-}{B}{C}
\ncline{-}{B}{D}
\ncline{-}{C}{D}
\ncline{-}{B}{E}
\ncline{-}{C}{E}
\ncline{-}{D}{E}
\ncline{-}{B}{F}
\ncline{-}{C}{F}
\ncline{-}{D}{F}
\ncline{-}{E}{F}
\end{pspicture}
}
}
\caption{Connected graphs of order $N=6$ and size $5\leq m \leq 15$ which have the smallest amount of overall crosstalk sensitivity among their vertices, when the intensity of the crosstalk interaction between two vertices is equal to the sum of their individual affinities, and the affinity of each vertex is proportional to the square root of its degree $f(d)=c\sqrt{d}$. Note that we cannot find the optimal graph of size $m$ recursively, by adding one extra edge to the optimal graph of size $m-1$.}
\label{All6Additive}
\end{figure}

\begin{figure}[htbp]
\centering
\subfigure{
\includegraphics[scale=0.35]{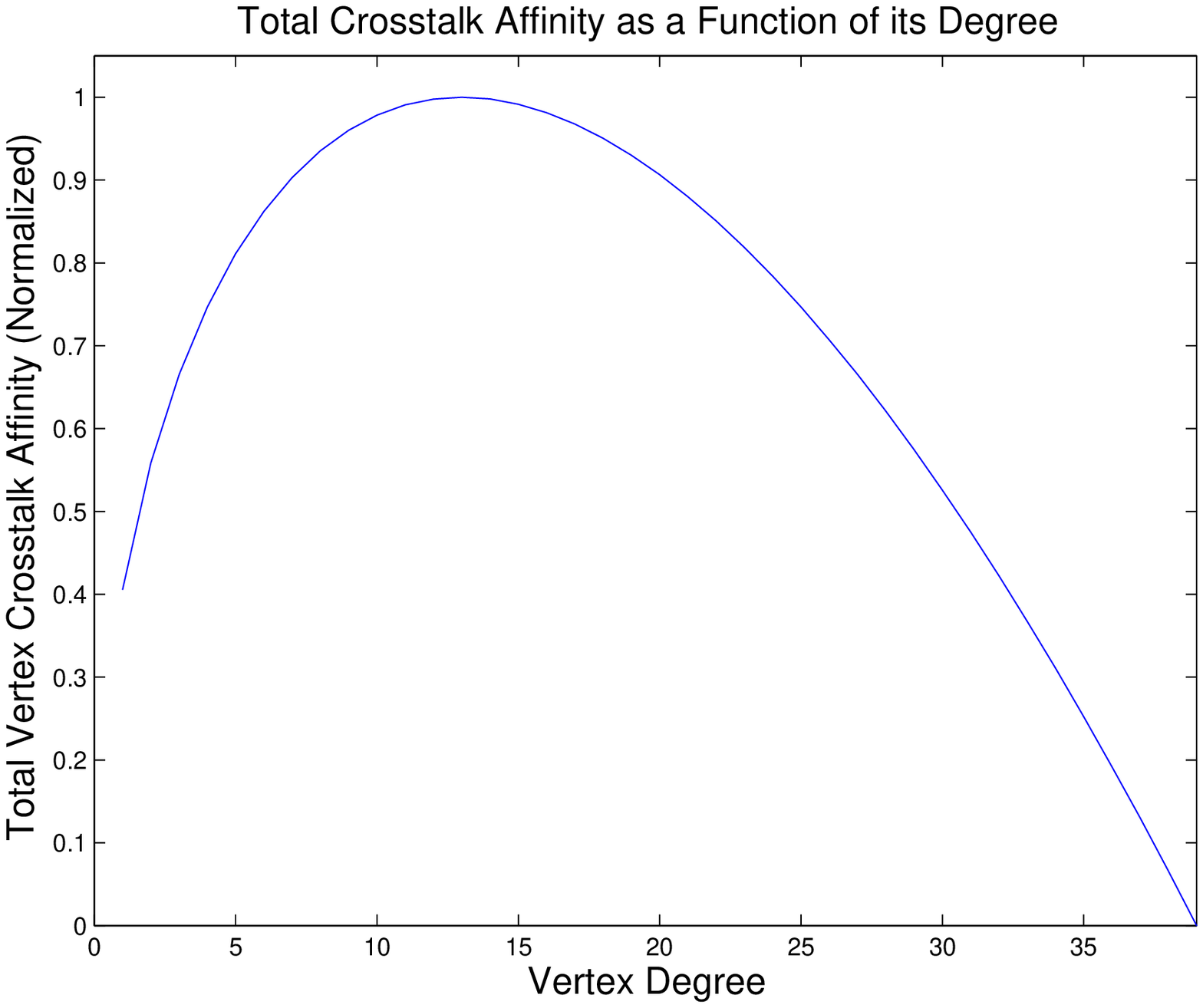}
}
\subfigure{
\includegraphics[scale=0.35]{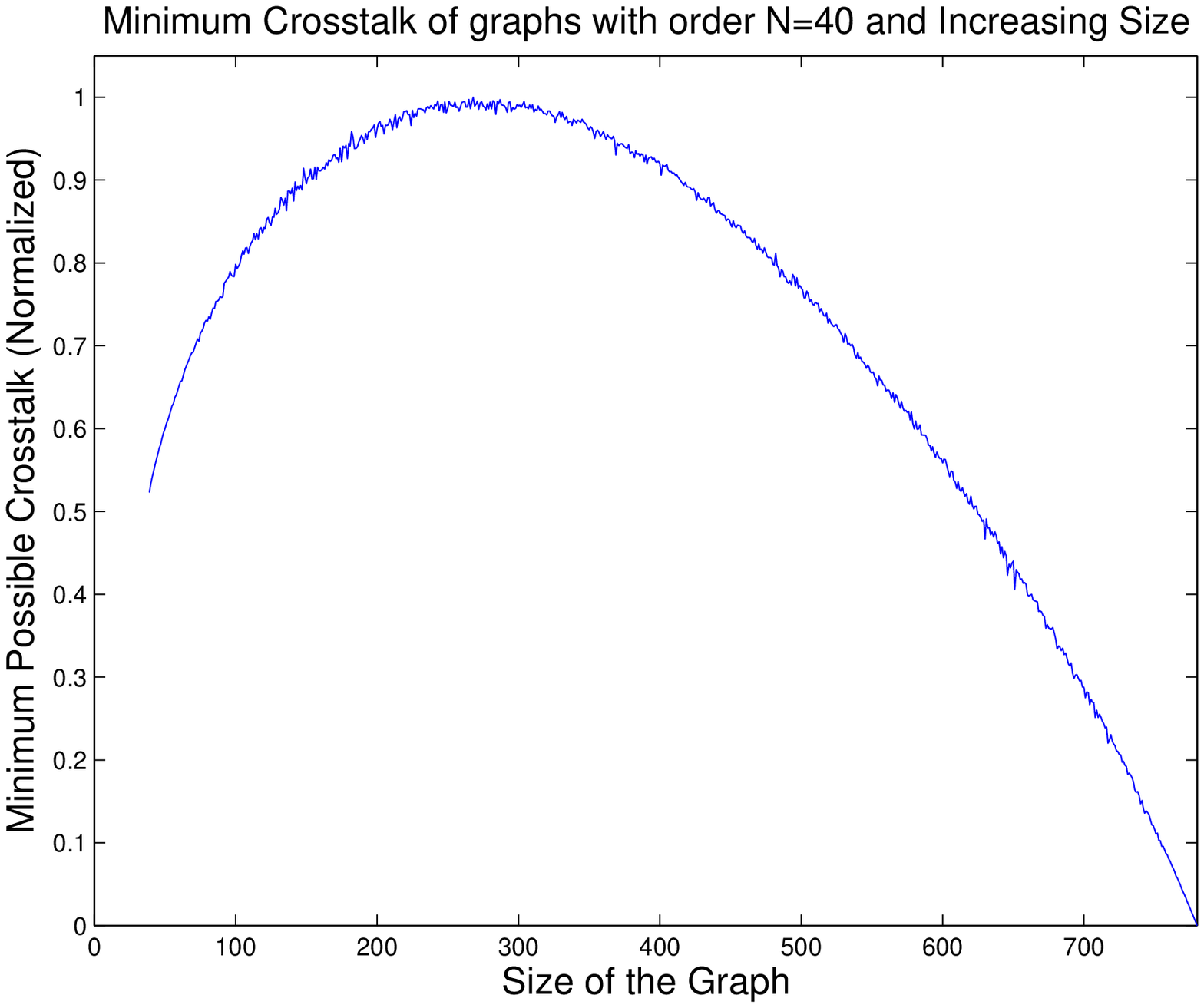}
}
\caption{\textbf{(a)} An example of the overall vertex affinity function $h$. Here, $f(d)=c\sqrt{d}$, and $h(d)=(N-1-d)f(d)$. \textbf{(b)} The minimum possible overall crosstalk  sensitivity of all graphs with $N=40$ vertices and $39\leq m \leq 780$ edges has the shape of the vertex crosstalk function. It is not smooth because $\mathbf{d}^{*}$ needs to satisfy additional constraints, most notably $\mathbf{d}^{*} \in \mathcal{P}$.}
\label{AdditiveXtalkEvolutionExample}
\end{figure}

Figure \ref{All6Additive} shows an example all the connected graphs with $N=6$ vertices which have the minimum additive crosstalk sensitivity.
Figure \ref{AdditiveXtalkEvolutionExample} shows the form of the total crosstalk affinity of each vertex, and the overall crosstalk of a graph of order $N=40$ as its size increases.

\FloatBarrier

%%%%%%%%%%%%%%%%%%%%%%%%%%%%%%%%%%%%%%%%%%%%%%%%
%%%%%%%%%%%%%%%%%%%%%%%%%%%%%%%%%%%%%%%%%%%%%%%%
%%%%%%%%%%%%%%%%%%%%%%%%%%%%%%%%%%%%%%%%%%%%%%%%
\section{Multiplicative Affinities}

If the crosstalk between two vertices is proportional to the product of each node's affinity, then we can write the pairwise affinity function $g(x,y)$ as
\begin{equation}
g(x,y)=f(x)f(y).
\end{equation}
We will refer to this type of crosstalk interactions in the network as \textit{multiplicative crosstalk}.
In this case, the overall crosstalk sensitivity will be
\begin{align}
K &= \sum _{(u,v) \notin \mathcal{E}(\mathcal{G}) } g(d_{u},d_{v}) \notag \\
   &=\sum _{(u,v) \notin \mathcal{E}(\mathcal{G}) } f(d_{u})f(d_{v}).
\label{DefinitionMultiplicativeEquation}
\end{align}

Once again, we consider the vertex affinity function $f$ to be a positive, strictly increasing and concave function of the degree of each vertex.
We will first find the networks with the minimum crosstalk sensitivity under the assumption above and without any additional constraints.
We will then describe the connectivity of the networks with a fixed degree sequence which have the lowest sensitivity to multiplicative crosstalk.

%%%%%%%%%%%%%%%%%%%%%%%%%%%%%%%%%%%%%%%%%%%%%%%%
\subsection{Structure of Networks with Minimum Crosstalk Sensitivity }

We first define a common graph pattern, and prove some lemmas that will be needed in order to find the form of the graphs with the lowest possible sensitivity to crosstalk.

\begin{definition}
A graph of order $N$ and size ${N-1\choose 2}+1 \leq m < {N\choose 2}$ is called a \textit{Type I almost complete graph} when it consists of a clique of order $N-1$, and one additional vertex of degree $\alpha=m-{N-1\choose 2}$ connected to it \cite{OptimalClustering}. The vertices of the clique are called \textit{central vertices}, and the additional node is called a \textit{peripheral vertex}.
\end{definition}
An example of such a graph is shown in Figure \ref{AlmostCompleteGraph}.
Since the only type of almost complete graph that we will study is Type I, we will omit mentioning the type of the almost complete graph without danger of confusion.

\begin{center}
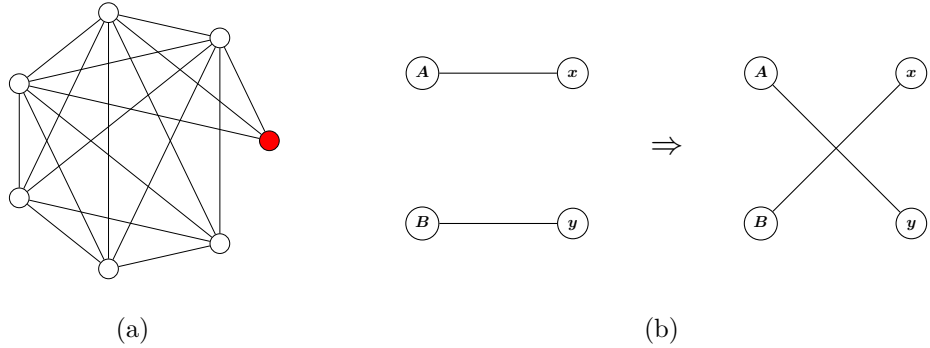
\begin{figure}[htbp]
\subfigure[]{
\psscalebox{0.35}{
\begin{pspicture}(-8,-6)(8,6)
{
\cnodeput(3.12,3.91){A}{\strut}
\cnodeput(-1.11,4.87){B}{\strut}
\cnodeput(-4.5,2.17){C}{\strut}
\cnodeput(-4.5,-2.17){D}{\strut}
\cnodeput(-1.11,-4.87){E}{\strut}
\cnodeput(3.12,-3.91){F}{\strut}
\cnodeput[fillstyle=solid,fillcolor=red](5,0){G}{\strut}
}
\ncline{-}{A}{B}
\ncline{-}{A}{C}
\ncline{-}{A}{D}
\ncline{-}{A}{E}
\ncline{-}{A}{F}
\ncline{-}{B}{C}
\ncline{-}{B}{D}
\ncline{-}{B}{E}
\ncline{-}{B}{F}
\ncline{-}{C}{D}
\ncline{-}{C}{E}
\ncline{-}{C}{F}
\ncline{-}{D}{E}
\ncline{-}{D}{F}
\ncline{-}{E}{F}
\ncline{-}{A}{G}
\ncline{-}{B}{G}
\ncline{-}{C}{G}
\end{pspicture}
}
\label{AlmostCompleteGraph}
}
\subfigure[]{
\psscalebox{0.5}{
\begin{pspicture}(-1,-2)(14,6)
{
\cnodeput(0,0){A}{\strut\boldmath$B$}
\cnodeput(0,4){B}{\strut\boldmath$A$}
\cnodeput(4,0){C}{\strut\boldmath$y$}
\cnodeput(4,4){D}{\strut\boldmath$x$}
\cnodeput(9,0){E}{\strut\boldmath$B$}
\cnodeput(9,4){F}{\strut\boldmath$A$}
\cnodeput(13,0){G}{\strut\boldmath$y$}
\cnodeput(13,4){H}{\strut\boldmath$x$}
}
\rput(6.5,2){\Huge $\Rightarrow$}
\ncline{-}{A}{C}
\ncline{-}{B}{D}
\ncline{-}{F}{G}
\ncline{-}{E}{H}
\end{pspicture}
}
\label{FirstTransformation}
}
\caption{\textbf{(a)} The \textit{Type I almost complete graph} consists of a clique of order $N-1$, and one peripheral vertex (shown in red). \textbf{(b)} If $d_{A}<d_{B}$ and $d_{x}>d_{y}$ then rewiring the edges $A-x$ and $B-y$ as shown will decrease the overall crosstalk sensitivity of the graph without changing its degree sequence.}
\label{AlmostComplete+Transformation}
\end{figure}
\end{center}

\begin{lemma}
In a graph with the minimum sensitivity to multiplicative crosstalk, if $A,B,x,y \in \mathcal{V}(\mathcal{G})$, then
\begin{equation}
d_{A}\leq d_{B} \quad \implies \quad d_{x}\leq d_{y} \qquad \forall x \in \mathcal{N}_{A}\cap \mathcal{N}^{^{c}}_{B} \quad \textrm{ and } \quad  y \in \mathcal{N}_{B}.
\label{MultipicativeXtalkNeighborDegrees}
\end{equation}
\label{MultiplicativeDegreeCorrelations}
\end{lemma}
\begin{proof}
If $\mathcal{N_{A}} \subseteq \mathcal{N}_{B}$ then the result holds trivially, since $\mathcal{N}_{A}\cap \mathcal{N}^{^{c}}_{B}=\emptyset$.
Otherwise, there is at least one vertex $x$ that is connected to $A$ and not connected to $B$. 
If we rewire the edges $(A,x)$ and $(B,y)$ to $(A,y)$ and $(B,x)$ as shown in Figure \ref{FirstTransformation}, the degree distribution stays the same for every vertex in the graph.
The difference in the overall sensitivity is 
\begin{align}
\Delta K =&f(d_{A})f(d_{x})+f(d_{B})f(d_{y})-f(d_{A})f(d_{y})-f(d_{B})f(d_{x}) \notag \\
=&(f(d_{B})-f(d_{A}))(f(d_{y})-f(d_{x})).
\end{align}
This difference should be nonnegative if the graph was optimal before the rewiring.
Since $d_{A}\leq d_{B}$, the sensitivity difference above can only be nonnegative if $d_{x}\leq d_{y}$.
\end{proof}

\begin{lemma}
Assume that the graph $\mathcal{G}_{N,m}$ with vertex set $\mathcal{V}$ and edge set $\mathcal{E}$ has the minimum sensitivity to multiplicative crosstalk and that the vertex affinity $f:{\Bbb R}^{+} \to {\Bbb R}^{+}$ is a strictly increasing concave function. Then
\begin{equation}
(a,b) \in \mathcal{E} \implies (a,x) \in \mathcal{E} \quad \forall a \in \mathcal{V} \quad \textrm{and} \quad \forall x \in \mathcal{V}, \textrm{such that} \quad d_{x}\geq d_{b}.
\label{MultiplicativeConnectToLargestDegreesEq}
\end{equation}
\label{MultiplicativeConnectToLargestDegrees}
\end{lemma}
\begin{proof}
Suppose that this is not the case, and there exists a vertex $a$ and two vertices $b$ and $c$ such that
\begin{equation}
(a,b) \in \mathcal{E}, \quad (a,c) \notin \mathcal{E} \quad \textrm{with} \quad d_{b}<d_{c}.
\end{equation}
Then, by rewiring the edge between $a$ and $b$ so that it now connects vertices $a$ and $c$, the difference in the crosstalk sensitivity will be
\begin{align}
\Delta K=&K_{new}-K_{old} \notag \\
=&f(d_{a})(f(d_{b})-f(d_{c})) \notag \\
&\quad+(f(d_{b}-1)-f(d_{b}))\sum _{u \in \mathcal{N}^{c}_{b-a}}f(d_{u}) \\
&\quad+ (f(d_{c}+1)-f(d_{c}))\sum _{v \in \mathcal{N}^{c}_{c-a}}f(d_{v}). \notag
\end{align}
Since $f$ is strictly increasing,
\begin{equation}
f(d_{b})<f(d_{c}), \quad f(d_{b}-1)-f(d_{b})<0 \quad \textrm{and} \quad f(d_{c}+1)-f(d_{c})>0
\end{equation} 
and from the concavity of $f$ we see that 
\begin{equation}
|f(d_{b}-1)-f(d_{b})|>|f(d_{c}+1)-f(d_{c})|.
\end{equation}
In addition, from Lemma \ref{MultiplicativeDegreeCorrelations},
\begin{equation}
\sum _{u \in \mathcal{N}_{b-a}}f(d_{u})<\sum _{v \in \mathcal{N}_{c-a}}f(d_{v}).
\end{equation}
Taking into account that the degrees of all other vertices in $\mathcal{G}$ remain constant after the transformation, 
\begin{equation}
\sum _{\substack{w \in \mathcal{V}(\mathcal{G})\\ w \neq b, w\neq c}}f(d_{w})=\textrm{constant}
\end{equation}
we find that
\begin{equation}
\sum _{u \in \mathcal{N}^{c}_{b-a}}f(d_{u})>\sum _{v \in \mathcal{N}^{c}_{c-a}}f(d_{v}).
\end{equation}
Combining the above equations, 
\begin{equation}
\Delta K <0.
\end{equation}
This means that the overall crosstalk sensitivity of the graph $\mathcal{G}$  has decreased after applying this transformation, which is a contradiction.
\end{proof}
The lemma above shows a necessary condition for the optimality of a network with respect to crosstalk: Every vertex should be connected with vertices of the largest degree possible.

\begin{corollary}
In the network with the minimum sensitivity to multiplicative crosstalk, the neighborhood set of every vertex $u$ with degree $d_{u}$ is 
\begin{equation}
\mathcal{N}_{u}=\{ N-d_{u}+1, N-d_{u}+2, \dots ,  N \} .
\label{MultiplicativeXtalkExplicitNeighbors}
\end{equation}
\label{MultiplicativeXtalkVertexNeighborhoods}
\end{corollary}

\begin{corollary}
In a network with the minimum sensitivity to multiplicative crosstalk, there is at least one vertex with full degree.
\label{MultiplicativeXtalkFullDegreeVertex}
\end{corollary}
\begin{proof}
From Corollary \ref{MultiplicativeXtalkVertexNeighborhoods}, $u$ will connect to the $d_{u}$ vertices with the largest degrees.
In addition, every vertex $u$ has degree $d_{u} \geq 1$, since we assumed that the graph is connected.
As a result, the vertex with the largest degree in the network is connected to all other $N-1$ vertices in the graph.
\end{proof}

Now we are ready to prove the main theorem of this section.
\begin{theorem}
If the crosstalk affinity of each vertex is a positive, increasing and concave function of its degree, then a graph with minimum overall multiplicative crosstalk sensitivity consists of a complete or almost complete subgraph, and the vertices that do not belong to it have degree equal to $1$. All vertices of degree $1$ have the same common neighbor, which is a vertex with full degree.
\end{theorem}
\begin{proof}
We will start from any graph $\mathcal{G}$ that satisfies equations  (\ref{MultipicativeXtalkNeighborDegrees}) and (\ref{MultiplicativeConnectToLargestDegreesEq}) but does not have the aformentioned form.
Then, with transformations that only reduce the sensitivity to crosstalk, we will arrive at the graph with this structure.
Suppose that the initial graph $\mathcal{G}$ has a degree distribution 
\begin{equation}
\mathbf{d}=[d_{1}, d_{2}, \dots , d_{N}]
\end{equation}
where the elements of the vertex set $\mathcal{V}(\mathcal{G})=\{1,2,\dots N \} $ are ordered in increasing degree, such that 
\begin{equation}
1\leq d_{1}\leq d_{2} \leq \dots \leq d_{N} \leq N-1.
\label{VertexDegreesInOrder}
\end{equation}
The overall sensitivity to crosstalk can be written as
\begin{align}
K(\mathcal{G})=&\sum _{(u,v) \notin \mathcal{E}(\mathcal{G})} f(d_{u})f(d_{v}) \notag \\
=& \sum _{u \in \mathcal{V}(\mathcal{G})}f(d_{u}) \Bigg [ \sum _{\substack{v \in \mathcal{V}(\mathcal{G})\\ v > u, v \notin \mathcal{N}_{u} }} f(d_{v}) \Bigg ] .
\end{align}
Applying the above equation to our specific $\mathcal{V}(\mathcal{G})$, 
\begin{equation}
K(\mathcal{G})=\sum _{u=1}^{N} f(d_{u}) \Bigg [ \sum _{\substack{v=u+1\\ v \notin \mathcal{N}_{u} }}^{N} f(d_{v}) \Bigg ] .
\end{equation}
Taking into account Corollary \ref{MultiplicativeXtalkVertexNeighborhoods}, we can simplify the expression of the overall crosstalk sensitivity as
\begin{equation}
K(\mathcal{G})=\sum _{u=1}^{N} f(d_{u}) \Bigg [ \sum _{v=u+1}^{N-d_{u}} f(d_{v}) \Bigg ] .
\end{equation}
Now, suppose that $\mathbf{d}$ is not a degree sequence that corresponds to the form in the statement of the theorem.
We pick the vertex $a$ with the smallest possible degree such that $d_{a}>1$.
We also pick vertex $b=a+1$, with degree $d_{b}\geq d_{a}$.
From equation (\ref{MultiplicativeXtalkExplicitNeighbors}), we get
\begin{equation}
\mathcal{N}_{a} \subseteq \mathcal{N}_{b}.
\end{equation}
More precisely,
\begin{equation}
\mathcal{N}_{a}=\{N-d_{a}+1, N-d_{a}+2, \dots ,N\} \quad \textrm{and} \quad \mathcal{N}_{b}=\{N-d_{b}+1, N-d_{b}+2, \dots ,N\}.
\end{equation}
By successively rewiring edges in $\mathcal{G}$, we will change the overall sensitivity each time by $\Delta K$, and when the graph is transformed to $\mathcal{G}^{*}$ with 
\begin{equation}
K(\mathcal{G}^{*}) \leq K(\mathcal{G})
\end{equation}
no additional transformations will be possible.
If $\mathcal{G}$ does not have the form in the statement of the theorem, there is at least one edge in the graph that can be rewired from vertex $a$ to vertex $b$.
We delete the edge $(a,N+1-d_{a})$, and add an edge to connect vertices $b$ and $N-d_{b}$.
It is worth noting that according to Corollary \ref{MultiplicativeXtalkFullDegreeVertex}, any such transformation will keep the graph connected, because the vertex with full degree is connected to every other vertex in the network as long as we do not remove any edge from a vertex with degree $1$.
The size of the graph remains the same, and the new neighborhoods of vertices $a$ and $b$ are respectively
\begin{equation}
\mathcal{N}_{a}'=\{N-d_{a}+2, N-d_{a}+3 \dots N\} \quad \textrm{and} \quad \mathcal{N}_{b}'=\{N-d_{b}, N-d_{b}+1, \dots N\}.
\end{equation}
We denote  the vertices with the smallest degrees  that are neighbors to $b$ and $a$ as $C=N+1-d_{b}$ and $D=N+1-d_{a}$ respectively.
According to inequality (\ref{VertexDegreesInOrder}),
\begin{equation}
C\leq D \implies d_{C}\leq d_{D}.
\end{equation}
The difference in crosstalk after applying the transformation is 
\begin{align}
\Delta K =&f(d_{a}-1) \sum_{u=a+1}^{N+1-d_{a}}f(d_{u}) + f(d_{b}+1)\sum_{v=b+1}^{N-1-d_{b}}f(d_{v}) \notag \\
&\quad -f(d_{a}) \sum_{u=a+1}^{N-d_{a}}f(d_{u}) + f(d_{b})\sum_{v=b+1}^{N-d_{b}}f(d_{v}) \notag \\
&\quad +(f(d_{C}+1)+f(d_{D}-1))\sum_{z=1}^{a-1}f(1) \\
&\quad -(f(d_{C})+f(d_{D}))\sum_{z=1}^{a-1}f(1). \notag
\end{align}
Rearranging the terms, we get
\begin{align}
\Delta K =&[f(d_{a}-1)-f(d_{a})]\Bigg [\sum_{u=a+1}^{N-d_{a}}f(d_{u})  \Bigg ] \notag \\
&\quad + [f(d_{b}+1)-f(d_{b})] \Bigg [ \sum_{v=b+1}^{N-d_{b}}f(d_{v}) \Bigg ]  \\
&\quad +[f(d_{a}-1)f(d_{D})-f(d_{b}+1)f(d_{C})]  \notag  \\
&\quad +[f(d_{D}-1)-f(d_{D}) +f(d_{C}+1)-f(d_{C}) ] \Bigg [ \sum_{z=1}^{a-1}f(1)  \Bigg ] . \notag
\end{align}
The sum above is clearly negative, so the transformation reduces the overall sensitivity to crosstalk of graph $\mathcal{G}$.
We may repeat the process described here, reducing the overall sensitivity at each step, until we transform the initial graph to the graph described in the theorem statement.
\end{proof}

Figure \ref{All6Multiplicative} shows all the graphs of order $N=6$ and size $5\leq m \leq 15$ with the minimum overall crosstalk sensitivity. 
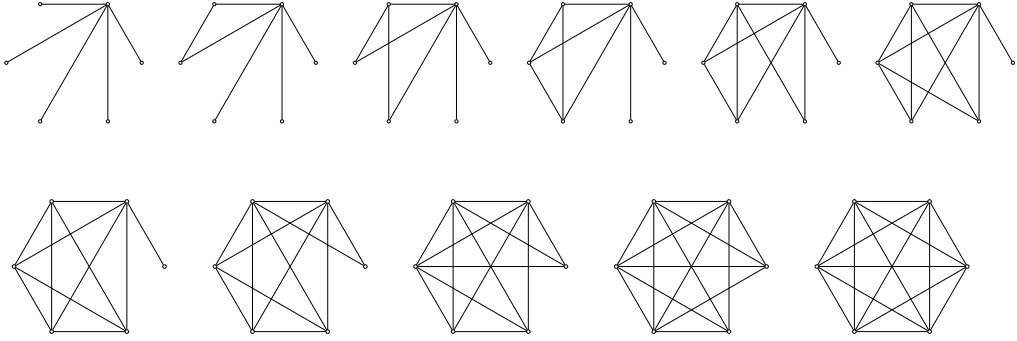
\begin{figure}[htbp]
\subfigure{
\psscalebox{0.18}{
\begin{pspicture}(-5,-5)(5,6)
{
\cnodeput(2.5,4.33){A}{}
\cnodeput(-2.5,4.33){B}{}
\cnodeput(-5,0){C}{}
\cnodeput(-2.5,-4.33){D}{}
\cnodeput(2.5,-4.33){E}{}
\cnodeput(5,0){F}{}
}
\ncline{-}{A}{B}
\ncline{-}{A}{C}
\ncline{-}{A}{D}
\ncline{-}{A}{E}
\ncline{-}{A}{F}
\end{pspicture}
}
}
\subfigure{
\psscalebox{0.18}{
\begin{pspicture}(-5,-5)(5,8)
{
\cnodeput(2.5,4.33){A}{}
\cnodeput(-2.5,4.33){B}{}
\cnodeput(-5,0){C}{}
\cnodeput(-2.5,-4.33){D}{}
\cnodeput(2.5,-4.33){E}{}
\cnodeput(5,0){F}{}
}
\ncline{-}{A}{B}
\ncline{-}{A}{C}
\ncline{-}{A}{D}
\ncline{-}{A}{E}
\ncline{-}{A}{F}
\ncline{-}{B}{C}
\end{pspicture}
}
}
\subfigure{
\psscalebox{0.18}{
\begin{pspicture}(-5,-5)(5,8)
{
\cnodeput(2.5,4.33){A}{}
\cnodeput(-2.5,4.33){B}{}
\cnodeput(-5,0){C}{}
\cnodeput(-2.5,-4.33){D}{}
\cnodeput(2.5,-4.33){E}{}
\cnodeput(5,0){F}{}
}
\ncline{-}{A}{B}
\ncline{-}{A}{C}
\ncline{-}{A}{D}
\ncline{-}{A}{E}
\ncline{-}{A}{F}
\ncline{-}{B}{C}
\ncline{-}{B}{D}
\end{pspicture}
}
}
\subfigure{
\psscalebox{0.18}{
\begin{pspicture}(-5,-5)(5,8)
{
\cnodeput(2.5,4.33){A}{}
\cnodeput(-2.5,4.33){B}{}
\cnodeput(-5,0){C}{}
\cnodeput(-2.5,-4.33){D}{}
\cnodeput(2.5,-4.33){E}{}
\cnodeput(5,0){F}{}
}
\ncline{-}{A}{B}
\ncline{-}{A}{C}
\ncline{-}{A}{D}
\ncline{-}{A}{E}
\ncline{-}{A}{F}
\ncline{-}{B}{C}
\ncline{-}{B}{D}
\ncline{-}{C}{D}
\end{pspicture}
}
}
\subfigure{
\psscalebox{0.18}{
\begin{pspicture}(-5,-5)(5,8)
{
\cnodeput(2.5,4.33){A}{}
\cnodeput(-2.5,4.33){B}{}
\cnodeput(-5,0){C}{}
\cnodeput(-2.5,-4.33){D}{}
\cnodeput(2.5,-4.33){E}{}
\cnodeput(5,0){F}{}
}
\ncline{-}{A}{B}
\ncline{-}{A}{C}
\ncline{-}{A}{D}
\ncline{-}{A}{E}
\ncline{-}{A}{F}
\ncline{-}{B}{C}
\ncline{-}{B}{D}
\ncline{-}{C}{D}
\ncline{-}{B}{E}
\end{pspicture}
}
}
\subfigure{
\psscalebox{0.18}{
\begin{pspicture}(-5,-5)(5,8)
{
\cnodeput(2.5,4.33){A}{}
\cnodeput(-2.5,4.33){B}{}
\cnodeput(-5,0){C}{}
\cnodeput(-2.5,-4.33){D}{}
\cnodeput(2.5,-4.33){E}{}
\cnodeput(5,0){F}{}
}
\ncline{-}{A}{B}
\ncline{-}{A}{C}
\ncline{-}{A}{D}
\ncline{-}{A}{E}
\ncline{-}{A}{F}
\ncline{-}{B}{C}
\ncline{-}{B}{D}
\ncline{-}{C}{D}
\ncline{-}{B}{E}
\ncline{-}{C}{E}
\end{pspicture}
}
}
\subfigure{
\psscalebox{0.2}{
\begin{pspicture}(-5.5,-5)(5.5,8)
{
\cnodeput(2.5,4.33){A}{}
\cnodeput(-2.5,4.33){B}{}
\cnodeput(-5,0){C}{}
\cnodeput(-2.5,-4.33){D}{}
\cnodeput(2.5,-4.33){E}{}
\cnodeput(5,0){F}{}
}
\ncline{-}{A}{B}
\ncline{-}{A}{C}
\ncline{-}{A}{D}
\ncline{-}{A}{E}
\ncline{-}{A}{F}
\ncline{-}{B}{C}
\ncline{-}{B}{D}
\ncline{-}{C}{D}
\ncline{-}{B}{E}
\ncline{-}{C}{E}
\ncline{-}{D}{E}
\end{pspicture}
}
}
\subfigure{
\psscalebox{0.2}{
\begin{pspicture}(-5.5,-5)(5.5,8)
{
\cnodeput(2.5,4.33){A}{}
\cnodeput(-2.5,4.33){B}{}
\cnodeput(-5,0){C}{}
\cnodeput(-2.5,-4.33){D}{}
\cnodeput(2.5,-4.33){E}{}
\cnodeput(5,0){F}{}
}
\ncline{-}{A}{B}
\ncline{-}{A}{C}
\ncline{-}{A}{D}
\ncline{-}{A}{E}
\ncline{-}{A}{F}
\ncline{-}{B}{C}
\ncline{-}{B}{D}
\ncline{-}{C}{D}
\ncline{-}{B}{E}
\ncline{-}{C}{E}
\ncline{-}{D}{E}
\ncline{-}{B}{F}
\end{pspicture}
}
}
\subfigure{
\psscalebox{0.2}{
\begin{pspicture}(-5.5,-5)(5.5,8)
{
\cnodeput(2.5,4.33){A}{}
\cnodeput(-2.5,4.33){B}{}
\cnodeput(-5,0){C}{}
\cnodeput(-2.5,-4.33){D}{}
\cnodeput(2.5,-4.33){E}{}
\cnodeput(5,0){F}{}
}
\ncline{-}{A}{B}
\ncline{-}{A}{C}
\ncline{-}{A}{D}
\ncline{-}{A}{E}
\ncline{-}{A}{F}
\ncline{-}{B}{C}
\ncline{-}{B}{D}
\ncline{-}{C}{D}
\ncline{-}{B}{E}
\ncline{-}{C}{E}
\ncline{-}{D}{E}
\ncline{-}{B}{F}
\ncline{-}{C}{F}
\end{pspicture}
}
}
\subfigure{
\psscalebox{0.2}{
\begin{pspicture}(-5.5,-5)(5.5,8)
{
\cnodeput(2.5,4.33){A}{}
\cnodeput(-2.5,4.33){B}{}
\cnodeput(-5,0){C}{}
\cnodeput(-2.5,-4.33){D}{}
\cnodeput(2.5,-4.33){E}{}
\cnodeput(5,0){F}{}
}
\ncline{-}{A}{B}
\ncline{-}{A}{C}
\ncline{-}{A}{D}
\ncline{-}{A}{E}
\ncline{-}{A}{F}
\ncline{-}{B}{C}
\ncline{-}{B}{D}
\ncline{-}{C}{D}
\ncline{-}{B}{E}
\ncline{-}{C}{E}
\ncline{-}{D}{E}
\ncline{-}{B}{F}
\ncline{-}{C}{F}
\ncline{-}{D}{F}
\end{pspicture}
}
}
\subfigure{
\psscalebox{0.2}{
\begin{pspicture}(-5.5,-5)(5.5,8)
{
\cnodeput(2.5,4.33){A}{}
\cnodeput(-2.5,4.33){B}{}
\cnodeput(-5,0){C}{}
\cnodeput(-2.5,-4.33){D}{}
\cnodeput(2.5,-4.33){E}{}
\cnodeput(5,0){F}{}
}
\ncline{-}{A}{B}
\ncline{-}{A}{C}
\ncline{-}{A}{D}
\ncline{-}{A}{E}
\ncline{-}{A}{F}
\ncline{-}{B}{C}
\ncline{-}{B}{D}
\ncline{-}{C}{D}
\ncline{-}{B}{E}
\ncline{-}{C}{E}
\ncline{-}{D}{E}
\ncline{-}{B}{F}
\ncline{-}{C}{F}
\ncline{-}{D}{F}
\ncline{-}{E}{F}
\end{pspicture}
}
}
\caption{Connected graphs of order $N=6$ and size $5\leq m \leq 15$ which have the smallest overall crosstalk sensitivity, when the pairwise crosstalk intensity between two vertices is equal to the product of the individual affinities, and the affinity of each vertex is a positive, increasing and concave function of its degree.}
\label{All6Multiplicative}
\end{figure}

\begin{theorem}
If the crosstalk affinity of each vertex is a positive, increasing and concave function $f$ of its degree, then the minimum possible overall multiplicative crosstalk sensitivity of a graph is equal to 
\begin{align}
K =&(d-1-\alpha)f(\alpha)f(d-2)+(N-d)f(\alpha)f(1) \notag \\
&\quad +(N-d)(\alpha-1)f(d-1)f(1) \notag \\
&\quad +(N-d)(d-1-\alpha)f(d-2)f(1) \\
&\quad +{N-d \choose 2}f^{2}(1)  \notag
\end{align}
where 
\begin{equation}
d= \Bigg\lceil { \frac{1}{2} \left( 3 + \sqrt{9+8m-8N} \right) } \Bigg\rceil
\end{equation}
is the number of vertices in the (almost) complete subgraph and 
\begin{equation}
\alpha=m-{d-1 \choose 2}-(N-d)
\end{equation}
is the degree of the vertex with the smallest number of neighbors in it.
\end{theorem}
\begin{proof}
The overall sensitivity will be calculated by computing intensity of the crosstalk interactions among the different vertex groups in the minimum crosstalk sensitivity graphs.
Suppose that the almost complete (or complete) subgraph consists of $d$ vertices, $d-1$ of which form a complete subgraph and the peripheral vertex of the subgraph has degree has degree $\alpha$.
The degree of the peripheral vertex  is allowed to be equal to $d-1$, in which case, we have a complete subgraph.
In addition, there are $N-d$ vertices with degree equal to $1$.
Counting all edges of the graph, we find that
\begin{equation}
m={d-1 \choose 2}+\alpha + (N-d).
\label{SumOfEdges}
\end{equation}
This equation needs to be solved for the variables $d$ and $\alpha$, under the conditions
\begin{equation}
\alpha \in \mathbb{N}^{*}, d \in \mathbb{N}^{*}
\end{equation}
and
\begin{equation}
2\leq \alpha \leq d-1.
\end{equation}
As mentioned above, when $\alpha=d-1$, we have a complete subgraph, otherwise we have an almost complete subgraph.
In order to find the order of the almost complete graph, we assume that the optimal graph consists of a complete subgraph of order $d_{0}$, and that $\alpha=d_{0}-1$. 
Then the equation that needs to be solved has only one variable $d_{0}$
\begin{equation}
{d_{0} \choose 2}+ (N-d_{0})=m
\end{equation}
which after rearranging the terms becomes a second order polynomial in $d_{0}$, namely
\begin{equation}
d^{2}_{0}-3d_{0}+(2N-2m)=0
\end{equation}
with roots
\begin{equation}
d_{0}=\frac{1}{2} \left( 3 \pm \sqrt{9+8m-8N} \right).
\end{equation}
One of the solutions will be zero or negative ($m\geq N$ for connected graphs that are not trees), so we need to pick the positive solution.
In addition, $d$ needs to be an integer.
Since we have artificially increased $\alpha$ in the above solution, $d\geq d_{0}$ and $d$ is the smallest integer that is larger than $d_{0}$, so we need to pick the ceiling of the positive real number from the above equation, 
\begin{equation}
d= \Bigg\lceil { \frac{1}{2} \left( 3 + \sqrt{9+8m-8N} \right) } \Bigg\rceil.
\end{equation}
Given the value of $d$, the value of $\alpha$ can be found from equation (\ref{SumOfEdges}).
Now that the structure of the optimal network is specified, we need to add up all the crosstalk contributions from all the vertices of the graph:
\begin{itemize}

\item Crosstalk between the peripheral vertex and the $d-1-\alpha$ vertices of degree $d-2$ of the almost complete graph that it is not connected to:
\begin{equation}
K_{1}=(d-1-\alpha)f(\alpha)f(d-2).
\end{equation}
\item Crosstalk between the peripheral vertex and the $N-d$ single-edge vertices:
\begin{equation}
K_{2}=(N-d)f(\alpha)f(1).
\end{equation}
\item Crosstalk among the single-edge vertices:
\begin{equation}
K_{3}={N-d \choose 2}f^{2}(1).
\end{equation}
\item Crosstalk among the vertices that are connected to the peripheral vertex and the single-edge vertices:
\begin{equation}
K_{4}=(\alpha-1)(N-d)f(d-1)f(1).
\end{equation}
In the equation above, the first term is $\alpha -1$ instead of $\alpha$ because the vertices of degree $1$ are all connected to one of these $\alpha$ vertices of the (almost) complete graph, and the crosstalk of this vertex is zero.
\item Crosstalk among the vertices of the almost complete graph that are not connected to the peripheral vertex, and the single-edge vertices:
\begin{equation}
K_{5}=(d-1-\alpha)(N-d)f(d-2)f(1).
\end{equation}
\end{itemize}

\begin{figure}[htbp]
\centering
\includegraphics[scale=0.55]{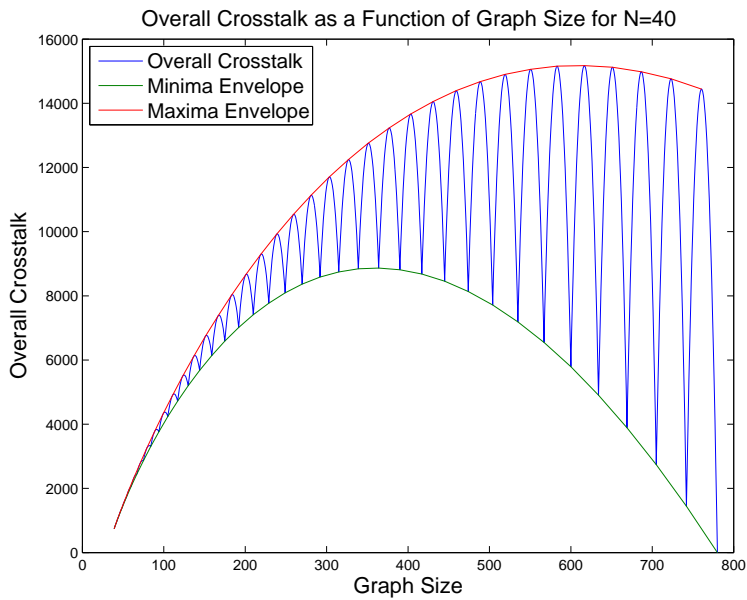} 
\caption{The minimum possible overall crosstalk sensitivity of a graph with $40$ vertices and a size of $39\leq m\leq 780$ when the vertex crosstalk affinity is $f(d)=d$. The green and red lines are the envelopes for the local minima and local maxima respectively.}
 \label{MultiplicativeOverallCrosstalkExample}
\end{figure}

Adding up the above terms, we get the total amount of crosstalk interactions in the network:
\begin{equation}
K=K_{1}+K_{2}+K_{3}+K_{4}+K_{5}.
\end{equation}
\begin{align}
K =&\underbrace{f(\alpha) \left[ (d-1-\alpha)f(d-2)+(N-d)f(1) \right]}_{\textrm{Total crosstalk of the peripheral vertex}} \notag \\
&+\underbrace{(N-d)f(1) \left[ (\alpha-1)f(d-1)+(d-1-\alpha)f(d-2) \right]}_{\textrm{Crosstalk between the almost complete graph and the single-edge vertices}} \label{TotalMultiplicativeXtalk} \\
&+ \underbrace{ {N-d \choose 2}f^{2}(1).}_{\textrm{Crosstalk among the single-edge vertices}} \notag
\end{align}
\end{proof}
An example of the overall crosstalk sensitivity as we increase the size of the graph of given order is shown in Figure \ref{MultiplicativeOverallCrosstalkExample}.

We notice that the overall crosstalk sensitivity has ``ripples''.
This happens because as we add more edges, the average degree of the vertices increases, and the average crosstalk intensity among vertices also becomes larger.
After adding one edge between two edges, the average decrease in crosstalk also increases, because there is no spurious interaction between them any more.
Consequently, the variance of the overall crosstalk sensitivity increases with number of edges in the network.
The local minima are achieved when the graph consists of a complete subgraph and the rest of the vertices have degree equal to one.

\subsection{Minimization of Crosstalk Sensitivity For Networks with Fixed Degree Sequence}
Assume that we have a network with a given degree distribution
\begin{equation}
\textbf{d}=[d_{1}, d_{2}, \dots, d_{N} ] 
\end{equation}
where the degree of each vertex $1\leq k \leq N$ is fixed and equal to $d_{k}$.
If we are free to choose the structure of the network, in order to minimize the overall crosstalk sensitivity of the network, we need to connect vertices with similar degrees together, as the next theorem shows.

\begin{theorem}
Suppose that we have a network with a given degree distribution $\mathbf{d}$.
Then, the structure that minimizes the sensitivity to multiplicative crosstalk is assortative in the vertex degrees.
\label{AssortativeDegreeDistributionTheorem}
\end{theorem}
\begin{proof}
Suppose we have a network, where vertices $A$ and $B$ are connected, vertices $C$ and $D$ are also connected, and there are no more edges among them.
One way to change the structure without changing the degree sequence of the graph is shown in Figure \ref{ConfigurationModel}.
This method (first described in \cite{MolloyReed1995}) keeps the degrees of each vertex constant.
\begin{figure}[htbp]
\centering
\psscalebox{.45}{
\begin{pspicture}(-1,-3)(15,15)
{
\cnodeput(0,8){A}{\strut\boldmath$A$}
\cnodeput(3,8){B}{\strut\boldmath$B$}
\cnodeput(0,5){C}{\strut\boldmath$C$}
\cnodeput(3,5){D}{\strut\boldmath$D$}
\cnodeput(10,13){E}{\strut\boldmath$A$}
\cnodeput(13,13){F}{\strut\boldmath$B$}
\cnodeput(10,10){G}{\strut\boldmath$C$}
\cnodeput(13,10){H}{\strut\boldmath$D$}
\cnodeput(10,3){I}{\strut\boldmath$A$}
\cnodeput(13,3){J}{\strut\boldmath$B$}
\cnodeput(10,0){K}{\strut\boldmath$C$}
\cnodeput(13,0){L}{\strut\boldmath$D$}
}
\ncline{-}{A}{B}
\ncline{-}{C}{D}
\ncline{-}{E}{H}
\ncline{-}{F}{G}
\ncline{-}{I}{K}
\ncline{-}{J}{L}
\pnode(4.2,7){T1start}
\pnode(8.5,11.5){T1end}
\pnode(4.2,6){T2start}
\pnode(8.5,1.5){T2end}
\pnode(11.5,8.5){T3start}
\pnode(11.5,4.5){T3end}

\ncline[linewidth=2pt]{->}{T1start}{T1end}
\rput(6,10){\rnode{TR1}{\Large{$T_{1}$}}}
\ncline[linewidth=2pt]{->}{T2start}{T2end}
\rput(6,3){\rnode{TR2}{\Large{$T_{2}$}}}
\ncline[linewidth=2pt]{->}{T3start}{T3end}
\rput(12.5,6.5){\rnode{TR3}{\Large{$T_{3}$}}}

\pspolygon[fillstyle=solid,linearc=14pt,cornersize=absolute,linewidth=1.5pt, opacity=0.3, linestyle=dashed]
(!\psGetNodeCenter{A} A.x 1 sub A.y 1 add)
(!\psGetNodeCenter{B} B.x 1 add B.y 1 add)
(!\psGetNodeCenter{D} D.x 1 add D.y 1 sub)
(!\psGetNodeCenter{C} C.x 1 sub C.y 1 sub)
\pspolygon[fillstyle=solid,linearc=14pt,cornersize=absolute,linewidth=1.5pt, opacity=0.3, linestyle=dashed]
(!\psGetNodeCenter{E} E.x 1 sub E.y 1 add)
(!\psGetNodeCenter{F} F.x 1 add F.y 1 add)
(!\psGetNodeCenter{H} H.x 1 add H.y 1 sub)
(!\psGetNodeCenter{G} G.x 1 sub G.y 1 sub)
\pspolygon[fillstyle=solid,linearc=14pt,cornersize=absolute,linewidth=1.5pt, opacity=0.3, linestyle=dashed]
(!\psGetNodeCenter{I} I.x 1 sub I.y 1 add)
(!\psGetNodeCenter{J} J.x 1 add J.y 1 add)
(!\psGetNodeCenter{L} L.x 1 add L.y 1 sub)
(!\psGetNodeCenter{K} K.x 1 sub K.y 1 sub)
\end{pspicture}
}
\caption{A rewiring method that keeps the degree of each vertex constant. Assume that vertices $A$ and $B$ are connected, as are $C$ and $D$. There are no other edges in this induced subgraph. We can rearrange the edges by connecting $A$ with $D$ and $B$ with $C$, or by connecting $A$ with $C$ and $B$ with $D$. Each transformation can be reversed, and we can go from any subgraph to another in one step.}
\label{ConfigurationModel}
\end{figure}
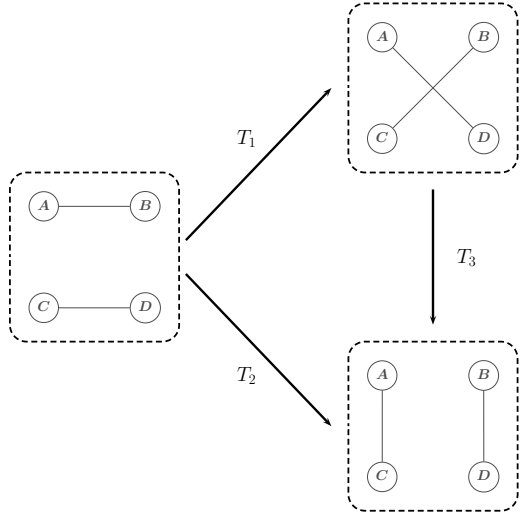
Initially, the overall crosstalk sensitivity among the four vertices $A, B, C$ and $D$ is
\begin{equation}
M_{0}=f(d_{A})f(d_{C})+f(d_{A})f(d_{D})+f(d_{B})f(d_{C})+f(d_{B})f(d_{D}).
\end{equation}
After transformation $T_{1}$ (see Figure \ref{ConfigurationModel} for details), the overall crosstalk sensitivity among the same vertices is 
\begin{equation}
M_{1}=f(d_{A})f(d_{B})+f(d_{A})f(d_{C})+f(d_{B})f(d_{D})+f(d_{C})f(d_{D})
\end{equation}
and similarly, if we apply the second transformation $T_{2}$
\begin{equation}
M_{2}=f(d_{A})f(d_{B})+f(d_{A})f(d_{D})+f(d_{B})f(d_{C})+f(d_{C})f(d_{D}).
\end{equation}
Consequently, the difference of the overall crosstalk sensitivity of the network if we apply $T_{1}$ is 
\begin{equation}
\Delta K_{T_{1}}=f(d_{A})f(d_{B})-f(d_{A})f(d_{D})-f(d_{B})f(d_{C})+f(d_{C})f(d_{D})
\end{equation}
\begin{equation}
\Delta K_{T_{1}}=\left(  f(d_{C}) -f(d_{A}) \right) \left(  f(d_{D}) -f(d_{B}) \right).
\end{equation}
If we apply the transformation $T_{2}$
\begin{equation}
\Delta K_{T_{2}}=f(d_{A})f(d_{B})-f(d_{A})f(d_{C})-f(d_{B})f(d_{D})+f(d_{C})f(d_{D})
\end{equation}
\begin{equation}
\Delta K_{T_{2}}=\left(  f(d_{C}) -f(d_{B}) \right) \left(  f(d_{D}) -f(d_{A}) \right).
\label{SecondConfModelTransformation}
\end{equation}
We order the degrees of the four vertices (there are $4!=24$ possible orderings), and without loss of generality, we require that 
\begin{equation}
d_{A}\leq d_{B} \quad \textrm{ and } \quad d_{C}\leq d_{D}
\end{equation}
which reduces the number of possible orderings to $6$, shown in Table \ref{All6Transformations}.
\begin{table}[htbp]
\begin{center}
\begin{tabular}{|c|c|c|}
\hline 
Degree Order & $\Delta K_{T_{1}}$ & $\Delta K_{T_{2}}$ \\
\hline
$d_{A} \leq d_{B} \leq d_{C} \leq d_{D}$ & $ \geq 0$ & $ \geq 0$ \\
\hline
$d_{A} \leq d_{C} \leq d_{B} \leq d_{D}$ & $ \leq 0$ & $ \leq 0$ \\
\hline
$d_{A} \leq d_{C} \leq d_{D} \leq d_{B}$ & $ \leq 0$ & $ \leq 0$ \\
\hline
$d_{C} \leq d_{A} \leq d_{B} \leq d_{D}$ & $ \leq 0$ & $ \leq 0$ \\
\hline
$d_{C} \leq d_{A} \leq d_{D} \leq d_{B}$ & $ \leq 0$ & $ \geq 0$ \\
\hline
$d_{C} \leq d_{D} \leq d_{A} \leq d_{B}$ & $ \geq 0$ & $ \geq 0$ \\
\hline
\end{tabular}
\caption{All possible orderings of the degrees of vertices $A$, $B$, $C$ and $D$, with $d_{A}\leq d_{B}$ and $d_{C}\leq d_{D}$, and the resulting difference in the overall crosstalk sensitivity among them for transformations $T_{1}$ and $T_{2}$ (see also Figure \ref{ConfigurationModel}).  }
\label{All6Transformations}
\end{center}
\end{table}
If both transformations result in an increase in the overall crosstalk sensitivity, this means that the current arrangement of edges ($(A,B)$ and $(C,D)$) is optimal. 
This happens when one of the following inequalities hold:
\begin{equation}
d_{A}\leq d_{B} \leq d_{C} \leq d_{D}
\end{equation}
\begin{equation}
d_{C}\leq d_{D} \leq d_{A} \leq d_{B}.
\end{equation}

Also note that transformation $T_{2}$ always yields a network with smaller overall crosstalk sensitivity than transformation $T_{1}$, because
\begin{align}
\Delta K_{T_{3}} =&M_{2}-M_{1} \notag \\
=&f(d_{A})f(d_{D})+f(d_{B})f(d_{C})-f(d_{A})f(d_{C})-f(d_{B})f(d_{D}) \notag \\
=&(f(d_{A})-f(d_{B}))(f(d_{D})-f(d_{C})) \\
\leq& 0 \notag
\end{align}
by the assumption of the relative degrees of the pairs $(A,B)$ and $(C,D)$.
More importantly, 
\begin{align}
M_{2} =&M_{0}+\Delta K_{T_{2}} \notag \\
=&M_{0}+ \Delta K_{T_{1}}+\Delta K_{T_{3}} \\
=&M_{1}+\Delta K_{T_{3}}. \notag
\end{align}
This shows that the overall sensitivity of a graph does not depend on the order in which we apply any transformations.
The difference in the overall crosstalk sensitivity is equal to the sum of individual differences,  and the only thing that matters is the structure of the network before and after the performed changes.
In order to minimize the overall sensitivity to crosstalk, we need to connect vertices with large degrees with other vertices with large degrees and similarly for vertices with small degrees.
\end{proof}

\begin{lemma}
A degree sequence $\mathbf{d}=[d_{1}, d_{2}, \dots d_{N}]$ is a graphic sequence if and only if the degree sequence $\mathbf{d}'=[d_{2}-1, d_{3}-1,\dots ,d_{d_{1}+1}-1,d_{d_{1}+2},\dots ,d_{N-1}, d_{N}]$ is a graphic sequence.
\end{lemma}
\begin{proof}
See \cite{DegreeRealizability}.
\end{proof}

\begin{algorithm}
\caption{Find the Connected Graph $\mathcal{G}$ with Minimum Sensitivity to Multiplicative Crosstalk}
\label{ConnectedGraphicalSequenceAlgorithm}
\begin{algorithmic}
\REQUIRE A graphic sequence $\mathbf{d}=[d_{1}, \dots d_{N}]$ with $d_{1}\leq d_{2}\leq \dots \leq d_{N}$.
\ENSURE The graph with the lowest sensitivity to crosstalk.
\STATE
\STATE $\mathcal{G} \leftarrow$ Empty Graph with $\mathcal{V}(\mathcal{G})$={1,2,\dots N}
\STATE $k \leftarrow 0$
\WHILE{$\mathbf{d} \neq [0]$}
\STATE $k \leftarrow k+1$
\STATE $r \leftarrow$ First element of $\mathbf{d}$
\STATE Connect vertex $k$ with vertices $k+1,k+2,\dots , k+r$
\STATE $\mathbf{d} \leftarrow [d_{2}-1, d_{3}-1, \dots , d_{r+1}-1, d_{r+2}, \dots, d_{N-1}, d_{N}]$
\ENDWHILE
\STATE
\WHILE{$\mathcal{G}$ is not connected}
\STATE Find the two components $\mathcal{F},\mathcal{S}$ with the smallest average degrees , $\bar{d}_{\mathcal{F}}\leq \bar{d}_{\mathcal{S}}$.
\STATE Pick vertices $A,B$ with the largest degree in $\mathcal{F}$ with $d_{A}\leq d_{B}$.
\STATE Pick vertices $C,D$ with the largest degree in $\mathcal{S}$ with $d_{C}\leq d_{D}$.
\STATE Add edges $(A,C)$ and $(B,D)$; Remove edges $(A,B)$ and $(C,D)$
\ENDWHILE
\STATE
\STATE Return $\mathcal{G}$
\end{algorithmic}
\end{algorithm}

\begin{theorem}
Algorithm \ref{ConnectedGraphicalSequenceAlgorithm} returns a connected network with the lowest possible overall sensitivity to multiplicative crosstalk.
\end{theorem}
\begin{proof}
By construction, the first part of the algorithm produces the network with the most assortative structure possible, and therefore the network with the smallest possible overall crosstalk sensitivity.
The only problem would be if this network is not connected.
In this case, we need to carry out one rewiring that will have minimum impact to the overall sensitivity.
It suffices to prove that the algorithm works when the graph has two unconnected components, since this procedure can easily be generalized for an arbitrary number of components.
Suppose that the first component $\mathcal{F}$ has degree sequence $d_{\mathcal{F}}=[d_{1}, d_{2}, \dots ,d_{c}]$ and the second component $\mathcal{S}$ has degree sequence $d_{\mathcal{S}}=[d_{c+1}, d_{c+2}, \dots ,d_{N}]$.
In order to make the graph connected, we need to apply the transformation $T_{2}$ of Figure \ref{ConfigurationModel}, to rewire one edge in $\mathcal{F}$ and one edge in $\mathcal{S}$ so that they connect vertices in the different components.
The difference of the overall crosstalk sensitivity is given in equation  (\ref{SecondConfModelTransformation}), where the four vertices have degrees
\begin{equation}
d_{c-1}\leq d_{c} \leq d_{c+1} \leq d_{c+2}.
\end{equation}
The difference will be
\begin{equation}
\Delta _{c}=(f(d_{c+1})-f(d_{c}))(f(d_{c+2})-f(d_{c-1})),
\end{equation}
which is the smallest possible positive difference while connecting the two components.
Repeating this procedure for all the components of the graph, we end up with a connected graph, which has the minimum possible additional overall sensitivity to crosstalk compared to the unconnected case.
\end{proof}

\FloatBarrier

\end{document}